\documentclass[11pt]{article}

\usepackage{amsmath}
\usepackage{fancyhdr}
\usepackage{amssymb}
\usepackage{amsthm}
\usepackage{graphicx}
\usepackage{varioref}
\usepackage{verbatim} 
\usepackage{multicol}
\usepackage{enumerate}
\usepackage[normalem]{ulem}
\usepackage{caption}
\usepackage{subcaption}
\usepackage[T1]{fontenc}
\usepackage[margin=2cm]{geometry}
\usepackage{mathtools}
\usepackage{algorithm, algpseudocode}

\usepackage{mathrsfs}
\usepackage[scr=esstix,cal=boondox]{mathalfa}

\usepackage{url}
\usepackage{xspace}
\usepackage{thm-restate}

 \usepackage{microtype}

\usepackage{tikz}
\usetikzlibrary{calc}

\usepackage{epstopdf}

\usepackage[colorlinks = true]{hyperref}
\usepackage{xcolor}
\definecolor{darkred}  {rgb}{0.5,0,0}
\definecolor{darkblue} {rgb}{0,0,0.5}
\definecolor{darkgreen}{rgb}{0,0.5,0}
\hypersetup{
  urlcolor   = blue,         %
  linkcolor  = darkblue,     %
  citecolor  = darkgreen,    %
  filecolor  = darkred       %
}

\usepackage[capitalise]{cleveref}
\crefname{lemma}{Lemma}{Lemmas}
\crefname{proposition}{Proposition}{Propositions}
\crefname{definition}{Definition}{Definitions}
\crefname{theorem}{Theorem}{Theorems}
\crefname{conjecture}{Conjecture}{Conjectures}
\crefname{corollary}{Corollary}{Corollaries}
\crefname{section}{Section}{Sections}
\crefname{appendix}{Appendix}{Appendices}
\crefname{figure}{Fig.}{Figs.}
\crefname{equation}{Eq.}{Eqs.}
\crefname{table}{Table}{Tables}
\crefname{claim}{Claim}{Claims}

\newtheorem{theorem}{Theorem}
\newtheorem{lemma}[theorem]{Lemma}

\newtheorem{definition}[theorem]{Definition}
\newtheorem{corollary}[theorem]{Corollary}

\newtheorem*{conjecture*}{Conjecture}

\theoremstyle{definition}

\usepackage{authblk}

\newcommand{\ket}[1]{|#1\rangle}
\newcommand{\bra}[1]{\langle#1|}
\newcommand{\proj}[1]{|#1\rangle\!\langle#1|}

\DeclareMathAlphabet{\matheu}{U}{eus}{m}{n}

\DeclareMathOperator{\tr}{tr}

\newcommand{\braket}[2]{\langle{#1}|{#2}\rangle}
\newcommand{\ketbra}[2]{|{#1}\rangle\!\langle{#2}|}

\newcommand{\C}{{c}} %
\newcommand{\ortho}{{w}} %
\newcommand{\SC}{{v}} %
\newcommand{\kap}{\kappa} %

\begin{document}

\title{Robust and Space-Efficient Dual Adversary Quantum Query Algorithms}
\author[1]{Michael Czekanski}
\author[2]{Shelby Kimmel}
\author[3]{R. Teal Witter}
\affil[1]{Cornell University Department of Statistics and Data Science \texttt{mc2589@cornell.edu}}
\affil[2]{Middlebury College Department of Computer Science \texttt{skimmel@middlebury.edu}}
\affil[3]{New York University Department of Computer Science and Engineering \texttt{rtealwitter@nyu.edu}}

\date{}

\maketitle

\abstract{
The general adversary dual is a powerful tool in quantum computing because it gives a query-optimal bounded-error quantum algorithm for deciding any Boolean function. Unfortunately, the algorithm uses linear qubits in the worst case, and only works if the constraints of the general adversary dual are exactly satisfied. The challenge of improving the algorithm is that it is brittle to arbitrarily small errors since it relies on a reflection over a span of vectors. We overcome this challenge and build a robust dual adversary algorithm that can handle approximately satisfied constraints. As one application of our robust algorithm, we prove that for any Boolean function with polynomially many 1-valued inputs (or in fact a slightly weaker condition) there is a query-optimal algorithm that uses logarithmic qubits. As another application, we prove that numerically derived, approximate solutions to the general adversary dual give a bounded-error quantum algorithm under certain conditions. Further, we show that these conditions empirically hold with reasonable iterations for Boolean functions with small domains. We also develop several tools that may be of independent interest, including a robust approximate spectral gap lemma, a method to compress a general adversary dual solution using the Johnson-Lindenstrauss lemma, and open-source code to find solutions to the general adversary dual.\footnote{Our repository is available at the following link: \url{https://github.com/rtealwitter/QuantumQueryOptimizer}.}

\section{Introduction}
 
The query model of computation has proven a powerful model in which
to prove quantum-classical separations
\cite{groverQuantumMechanicsHelps1997a,de2002sharp,childs2003exponential}
and to understand the limits of quantum algorithms \cite{ambainis2000quantum,beals2001quantum,hoyer2007negative}.
The power and usefulness of this computational model in the quantum setting
in part derives from the fact that bounded-error quantum query complexity for Boolean function evaluation is tightly and beautifully
described by a semidefinite program (SDP)---the general adversary bound \cite{reichardtSpanProgramsAre2014,reichardtReflectionsQuantumQuery2011}.  
The dual of this semidefinite program has
played an important role in the development and understanding of quantum 
algorithms.
In particular, the dual has been used to show the optimality of
span program algorithms, which are a critical element for several algorithm 
design paradigms \cite{reichardt2008span,beigiSpanProgramNonbinary2019,itoApproximateSpanPrograms2019,jeffery2023quantum} 
and are useful in a wide range of applications 
\cite{cadeTimeSpaceEfficient2018,delorenzoApplicationsQuantumAlgorithm2019a,carette2020extended,belovs2012learning,beigiQuantumSpeedupBased2019}.
While we have methods to create a
bounded-error quantum algorithm for function evaluation based on a set of
vectors that exactly satisfy the constraints of the general adversary dual \cite{reichardtReflectionsQuantumQuery2011},
it is natural to ask how robust this algorithm is to transformations and
perturbations. For example, can we take a vector set that satisfies the dual
adversary constraints and modify it to obtain an algorithm with better space complexity? Or
can we create an algorithm from a vector set that satisfies relaxed dual
adversary constraints? In this paper, we find criteria under which these
modified or approximately satisfying vector sets yield viable algorithms, and
we consider two problems where such robustness is useful.

First, we study when we can reduce the space complexity of
dual adversary-derived algorithms. Almost all Boolean functions on $n$ bits
have unitary space complexity $\Omega(n)$ \cite{jeffery2019span}, but we hope
to discover conditions under which less space might be required. We do this
by determining when we can compress the dimension of a set of vectors that exactly satisfies
the dual adversary constraints. While the dual adversary is most commonly
studied in the context of query complexity, is it also closely related to unitary
space complexity \cite{jeffery2019span}, and the dimension of the satisfying
vectors to the dual adversary problem determines the space used by
the algorithm. As building large quantum computers will likely continue to be
a technical challenge in the near to medium term \cite{IBM}, finding ways to
minimize the space used by quantum computers is important. 

We use two approaches to compress the dimension of a satisfying vector set,
and hence reduce the space used by the resulting algorithm. First, we
consider using a unitary transformation to rotate the vectors to a smaller
space. Next, we analyze applying the Johnson-Lindenstrauss (JL) lemma,
which is a powerful tool for compressing the dimension of a vector set while
approximately preserving the structure of the vectors. With our analysis,
we find the simpler, unitary transformation always results in a better
compression than the JL approach.
But, with either compression method, we can
show for any function with polynomially many 1-valued inputs, or polynomially
many 0-valued inputs, there is a query-optimal algorithm that uses
logarithmic space.
While it is not surprising that there is an algorithm that uses logarithmic space for
such problems (one could iterate through possible $1$/$0$-valued inputs and run Grover's search to test each one),
it is not obvious that there is a \textit{query-optimal} algorithm that uses logarithmic space.

The second problem we consider is how to create an algorithm using the output
of a numerical SDP solver applied to the general
adversary dual. One can plug the general adversary dual into a classical
SDP solver and find a set of vectors that is
close to a query-optimal, exactly satisfying vector set. 
Due to finite precision, we expect that the numerical solution will almost
never be \textit{exactly} satisfying, but we would like to know if we can still
produce a query-optimal quantum algorithm.
With our tools for creating robust dual adversary algorithms, 
we bound the error that can be tolerated, and we describe how to take a vector set that
does not quite satisfy the dual adversary constraints, but is within the
tolerable error, and use it to create a bounded-error algorithm. 
In the worst case, our analysis requires an error that scales
inversely with
the number of 1-valued or 0-valued inputs.
While this can be exponential in number of bits in the function, 
prior to our work, it was not clear how to create \textit{any} 
algorithm from an approximately satisfying solution.
Moreover, the general adversary bound is an SDP with dimension $|X|$
where $X$ is the function domain, so in general, solvers will already take time
polynomial in $|X|$ to solve classically \cite{vandenberghe1996semidefinite},
so we expect that attaining this level of precision will only contribute
polynomially to the overall runtime. Additionally,
we show numerically that, 
at least for small functions,
the error bound we require is easily attainable.

One may naturally wonder why a vector set that approximately satisfies the
dual adversary constraints does not immediately yield an appropriate
algorithm. The challenge is that the standard algorithm is based on a
reflection about the span of a subset of those vectors. For example, consider
a vector set that should ideally be $\{\ket{0},\ket{1},\ket{0}+\ket
{1}\}$, but is instead $\{\ket{0}+10^{-10}\ket{2},\ket{1},\ket{0}+\ket
{1}\}$. The span of the first set is 2-dimensional, but the span of the
second set is 3-dimensional, even though the two vector sets are very close
by almost any metric. If the algorithm reflects over a space that is much
larger than it should, it might not correctly evaluate the function on all
inputs. Our techniques allow us to find appropriate reflections (or in fact unitaries that are close to reflection) so that the
algorithms can proceed, even with errors like the example above.
Along the way towards proving our main results,
we also develop several tools that may be of independent interest, including a
robust approximate spectral gap lemma and 
open-source code to find solutions to the general adversary dual.

\subsection{Related Work}
\paragraph{Space Complexity and Compression} 
Reichardt observes that that the space used by
the dual adversary algorithm is the log of the rank of $Z$, where $Z$ is the
positive semidefinite matrix that optimizes the primal general adversary
bound, and he notes that this provides a worst case $\log(n|X|)$ space
complexity for $n$-bit functions with domain $X$
\cite{reichardtReflectionsQuantumQuery2011,reichardtSpanProgramsQuantum2009}. 
Our exact compression result (\cref{thm:exact_compress}) is of a similar
flavor, except that we are using the ``rank'' of the dual. Barnum,
Saks, and Szegedy use a different family of SDPs to characterize
query complexity \cite{barnum2003quantum} (these SDPs can give improved
performance in the case of small or zero error algorithms), and their
algorithm again depends on the rank of the satisfying positive semidefinite
matrix, but in the worst case uses $\log |X|+1$ qubits.

The key tool we use in our compression application is the Johnson-Lindenstrauss lemma.
The lemma guarantees that high-dimensional vectors randomly compressed into a lower-dimensional
space approximately preserves the inner products
of the vectors \cite{johnson1984extensions}.
The JL lemma is used in a variety of classical applications
including compressed sensing, dimensionality reduction,
and machine learning \cite{foucart2013invitation, vershynin2018high, carleo2019machine}, and it works even with sparse compression matrices
\cite{kane2014sparser}. In fact,
our work, in which we compress the dual solution to an SDP, 
has similarities to work by So et al. \cite{so2008unified}, which uses
JL compression to reduce the rank of the matrix that
is the primal solution to a semidefinite programming problem, at the cost of only approximately satisfying
the constraints.
It is also known that the compressed dimension given by the Johnson-Lindenstrauss lemma
is optimal up to constant factors \cite{larsen2017optimality}. 

The idea of relaxing SDP constraints in order to improve the space
used by an algorithm has also been considered in the classical regime.
Ding et al. create a storage-optimal SDP solver by relaxing
constraints \cite{ding2021optimal}.

In the quantum arena, a natural application of the JL lemma would be to compress the size of quantum states, but Harrow et al. found that there is no
such mapping that significantly reduces the dimension of quantum states while preserving
the Schatten 2-norm distance with high probability
\cite{harrow2015limitations}. However, the JL lemma was used to compress the
space used by a quantum finger printing protocol \cite{gavinsky2006strengths}.

Span programs (which are equivalent to the dual adversary \cite{reichardtReflectionsQuantumQuery2011,reichardtSpanProgramsAre2014}) were in fact originally formulated in order to understand
classical space complexity \cite{karchmer1993span}, and
Jeffery shows lower bounds on the space complexity of function evaluation
that depend on minimum span program and approximate span program sizes \cite{jeffery2019span}. 

\paragraph{Numerical Solutions to the Dual Adversary}
The idea of using classical computers
to design quantum algorithms is not new.
The variational quantum eigensolver
iteratively uses a classical computer to make a ground state ansatz, which is then tested by a quantum computer \cite{peruzzo2014variational}. A classical machine learning algorithm
can be used to guide quantum algorithm design \cite{bang2014strategy}. However,
the dual adversary semidefinite programming problem is different in that it automatically produces a query-optimal algorithm, rather than an iterative process
guided by classical, heuristic optimization methods.

\subsection{Open Questions}
Our techniques for space compression preserve the quantum query complexity 
of the original algorithm, while attempting to reduce space complexity. It 
would be very interesting if they could be modified to reduce space at the cost of 
increased query complexity; this might provide insight into one of Aaronson's 2021
open query complexity problems \cite{aaronson2021open}: better understanding space and query trade offs, specifically for the problems of collision and element 
distinctness.

While we analyze dual adversary algorithms, these are closely related to
span program algorithms. It should be possible to translate the bounds and 
conditions we find on the robustness of the general adversary dual into
analogous bounds and conditions on span program algorithms. We are especially
curious if these relaxed constraints could be related to Approximate Span Programs
\cite{itoApproximateSpanPrograms2019}, which are another way of relaxing 
the constraints of standard span programs.

While we show conditions under which it is possible to create dual adversary 
algorithms, we do not prove lower bounds. It would be interesting 
to study whether, with more detailed analysis, or under 
additional natural conditions, the JL
approach to space compression could be improved, or whether the analysis we present
is optimal. 

Generalizations of the general adversary bound
characterize the problems of quantum state conversion \cite{leeQuantumQueryComplexity2011} and quantum
subspace conversion \cite{belovs2023one}.
Perhaps our
techniques could be extended to these additional regimes.

\section{Preliminaries}\label{sec:background}
A few notational conventions: we use $\|\ket{\psi}\|$ do denote the $\ell_2$ norm of $\ket{\psi}$, $[n]$ to denote $\{1,2,\dots,n\}$, and $\delta_{i,j}$ to denote the Kronecker delta function. If $\ket{\lambda}$ is an eigenvector of $U$ with eigenvalue $e^{i\beta}$, we say 
the phase of $\ket{\lambda}$ is $\beta$.  For any unitary $U$, let
$P_\Theta(U)$ be the orthogonal projector onto the eigenvectors of $U$ with phase at
most $\Theta$. That is, $P_\Theta(U)$ is the orthogonal projector onto 
$\textrm{span}\{\ket{\lambda}:U\ket{\lambda}=e^{i\beta}\ket{\lambda}\textrm{ with }|\beta|\leq\Theta\}$. 

We consider the quantum query complexity and space complexity of evaluating a function 
$f:X\rightarrow\{0,1\}$ where 
$X\subseteq\{0,1\}^n$. For such a function $f$, we define $f^{-1}(b)=\{x\in X:f(x)=b\}.$
For some $x\in X\subseteq\{0,1\}^n$, we are given access to an oracle $O_x$ that 
acts on $\mathbb{C}^{n}\otimes \mathbb{C}^2$ as $O_x\ket{i}\ket{b}=\ket{i}\ket{b\oplus x_i}$, 
where $\ket{i}$ for $i\in [n]$ and $\ket{b}$ for $b\in\{0,1\}$ are standard basis 
states, and $x_i$ 
is the $i^\textrm{th}$ bit of $x$. Given $O_x$, we would like to determine $f(x)$.
We do this by implementing a bounded-error quantum query algorithm, which without loss of generality takes the form 
\begin{align}
U_TO_xU_{T-1}O_x\cdots U_1O_xU_0\ket{\hat{0}},
\end{align}
followed by a two-outcome measurement that determines the output of the algorithm, where $U_0,\dots U_T$ are unitary operations
acting on a Hilbert space of size $S$, such that for every
input $x\in X$, the probability of outputting $f(x)$ is at least $2/3$.
The algorithm uses $T$ applications of the oracle and $\log S$ qubits of space.
The bounded-error query complexity of $f$ is the minimum query complexity of any bounded-error query algorithm for $f$. 

The general adversary dual is used in designing query-optimal quantum algorithms for 
function evaluation: 
\begin{definition} [General Adversary Dual]\label{def:dual}
Let $f:X\rightarrow \{0,1\}$ for $X\subseteq\{0,1\}^n$. The following semidefinite 
optimization problem is called the dual of the general adversary bound, or what we 
call the \textit{general adversary dual}:
\begin{align}
&\min_{\substack{m\in \mathbb{N}\\\ket{v_{x,j}}\in \mathbb{C}^m}}\left\{\max_{x\in X}\sum_j\|\ket{v_{x,j}}\|^2\right\}\label{eq:filteredNormMa} \\
&\textrm{ s.t. }\forall x,y\in X:f(x)\neq f(y), \qquad 1=\sum_{j:x_j\neq y_j}\braket{v_{x,j}}{v_{y,j}}.\label{eq:filteredNorm}
\end{align}
\end{definition}

While \cref{def:dual} seeks to minimize the dimension $m$ of the vectors 
$\{\ket{v_{x,j}}\}_{x\in X,j\in[n]}$ 
(we will drop the set-building subscript and use $\{\ket{v_{x,j}}\}$ 
when clear from context) that also minimizes 
$\max_{x\in X}\sum_j\|\ket{v_{x,j}}\|^2$, we note that 
to design an algorithm, we only need a vector set $\{\ket{v_{x,j}}\}$ that satisfies the constraints in 
\cref{eq:filteredNorm}. This motivates the following definition, similar to converting vector sets in \cite{anderson2023improved}.
\begin{definition}[Deciding Vector Set and Related Terms]
Let $f:X\rightarrow \{0,1\}$ for $X\subseteq\{0,1\}^n$, and let $m\in \mathbb{N}$. 
Then $\{\ket{v_{x,j}}\in \mathbb{C}^m\}_{x\in X,j\in[n]}$ is an $f$-deciding vector 
set if
\begin{equation}\label{eq:DecideVectorConstraints}
\forall x,y\in X:f(x)\neq f(y), \qquad 1=\sum_{j:x_j\neq y_j}\braket{v_{x,j}}{v_{y,j}}.
\end{equation}
We say the {\em{size}} of $\{\ket{v_{x,j}}\}$ is $\max_{x\in X}\sum_j\|\ket{v_{x,j}}\|^2$, 
the {\em{dimension}} is $m$, and the {\em{maximum rank}} is $\max_{j \in [n]}  \quad
\operatorname{rank}\{\ket{v_{x,j}} : f(x) = 1\}$.
\label{def:decidingVecSet}
\end{definition}

Given an $f$-deciding vector set, one can design a query
algorithm to decide $f$:
\begin{theorem}[\cite{reichardtReflectionsQuantumQuery2011,leeQuantumQueryComplexity2011}]\label{thm:standard}
For $f:X\rightarrow\{0,1\}$ with $X\subseteq\{0,1\}^n$ let  $\{\ket{v_{x,j}}\}_{x\in X,j\in[n]}$ 
be an $f$-deciding vector set with size $A$ and dimension $m$. 
Then there is a bounded-error quantum query algorithm that decides $f$ with query complexity $O\left(A\right)$ and space complexity $O(\log(nm))$.
\end{theorem}
Because any $n$-bit function can be decided in $n$ queries, we assume $A=O(n).$ Additionally, applying Cauchy-Schwarz to \cref{eq:DecideVectorConstraints},
we have $A\geq 1.$

We sketch the proof of \cref{thm:standard} to make it easier to compare with our
algorithms. (For a more detailed description using similar notation, see 
Ref. \cite[Chapter 23.6]{ACNotes}). The subroutine used in both 
\cref{thm:standard} and in our algorithms is (parallelized) phase estimation. 

\begin{restatable}{lemma}{phaseEst}[Phase Estimation \cite{kitaevQuantumMeasurementsAbelian1995,cleveQuantumAlgorithmsRevisited1998,nagaj2009fast}]\label{lem:PhaseEst}
Let $U$ be a unitary that acts on $n$ qubits, and let $\delta,\Theta>0$. Then there 
is a phase estimation style circuit that has precision $\Theta$, error $\delta$, acts on $n+b$ qubits for $b=O\left(\log\frac{1}{\Theta}\log\frac{1}{\delta}\right)$, 
and
uses $O\left(\frac{1}{\Theta}\log\frac{1}{\delta}\right)$ calls to
control-$U$ applied to a single instance of the state $\ket{\psi}$, such that 
$p_0$, the probability of outcome $0$, satisfies 
\begin{equation}
\|P_{\Theta/2}(U)\ket{\psi}\|^2(1-\delta)-\delta\leq p_0\leq \|P_\Theta(U)\ket{\psi}\|^2+\delta.
\end{equation}
\end{restatable}
\noindent We prove \cref{lem:PhaseEst} in \cref{sec:prelimProofs}, which relies heavily on 
prior analyses of phase estimation circuits.

The algorithm of \cref{thm:standard} applies phase estimation with precision $O(1/A)$
 on a unitary $U$ with an initial state $\ket{\hat{0}}$.
$U$ acts on the space $\mathcal{H}=\mathbb{C}\oplus \mathbb{C}^n\otimes\mathbb{C}^m\otimes\mathbb{C}^2$, 
and is a product of two reflections: $U=(2\Pi_x-I)(2\Delta-I)$,
where $\Pi_x = \proj{\hat{0}}+\sum_{i\in[n]}\proj{i}\otimes I\otimes\proj{x_i}$ (here $I$ acts on $\mathbb{C}^m$, and $\ket{\hat{0}}$ is orthogonal to $\sum_{i\in[n]}\proj{i}\otimes I\otimes(\proj{0}+\proj{1})$)
and $\Delta$ is the orthogonal projector onto the following set of normalized vectors:
\begin{align}\label{eq:psidef}
\ket{\psi_y}&=\frac{1}{\sqrt{\nu_y}}\left(\ket{\hat{0}}+\frac{1}{\sqrt{\C A}}
\sum_{i\in[n]}\ket{i}\ket{v_{y,i}}\ket{y_i}\right) && &\forall y:f(y)=1,\\
\textrm{s.t. }\nu_y&=1+\frac{1}{\C A}\sum_{i\in[n]}\|\ket{v_{y,i}}\|^2 \leq 1+1/\C,\label{eq:nuxBound}
\end{align}
where $\nu_y\geq 1$ is chosen to normalize $\ket{\psi_y}$, and $\C$ is a constant chosen
depending on the desired success probability.
We note $(2\Pi_x-I)$ requires $2$ uses of $O_x$ to implement, and $(2\Delta-I)$
depends on the choice of the deciding vector set but is independent of the input $x$. 

Then when $f(x)=1$ and $\C=2$, $\ket{\hat{0}}$ has high overlap with $\ket{\psi_x}$, which is 
easily verified to be 
a $0$-phase eigenvector of $U$, so by \cref{lem:PhaseEst}, the probability of 
measuring a phase of $0$
when we perform phase estimation is large when $\delta$ (the error of phase estimation) is a small constant.

On the other hand, when $f(x)=0$, we consider the following normalized 
vector:
\begin{align}
\ket{\phi_x}&=\frac{1}{\sqrt{\mu_y}}\left(\ket{\hat{0}}-\sqrt{\C A}\sum_{i\in[n]}\ket{i}\ket{v_{x,i}}\ket{\bar{x}_i}\right) && 
\label{eq:phidef}
\\
\textrm{s.t. }\mu_x&=1+\C A\sum_{i\in[n]}\|\ket{v_{x,i}}\|^2\leq 1 + \C A^2, \label{eq:mu_x_bound}
\end{align}
where $\mu_y\geq 1$ is chosen to normalize the vector.
Because $\{\ket{v_{y,i}}\}$ is a deciding vector set, we have 
$\forall y:f(y)=1, \braket{\psi_y}{\phi_x}=0$. Thus $\{\ket{\phi_x}\}$ 
is orthogonal to $\Delta$. Next we use the effective spectral gap lemma:
\begin{restatable}{lemma}{approxGap}[Effective Spectral Gap Lemma \cite{leeQuantumQueryComplexity2011}]\label{lemma:approx_gap_standard}
Let $\Pi, \Delta$ be orthogonal projectors, let $U=(2\Pi-I)(2\Delta-I)$, and $\Delta\ket{w}=0$. Then
\begin{equation}
\|P_\Theta(U)\Pi\ket{w}\|\leq \Theta/2\|\ket{w}\|.
\end{equation}
\end{restatable}
Then the probability of measuring a phase of $0$ when we perform phase estimation when $f(x)=0$, by 
\cref{lem:PhaseEst} is upper bounded by a term that depends on 
$\|P_\Theta(U)\ket{\hat{0}}\|^2=\mu_y\|P_\Theta(U)\Pi_x\ket{\phi_x}\|^2$. 
Applying \cref{lemma:approx_gap_standard}, and using the fact that $\mu_y=O(A^2)$, 
we see this is
small when $\Theta$, the precision of phase estimation, is chosen to be $O(1/A)$. 
Since $\delta$ (the error of phase estimation) is chosen to be a small constant, by \cref{lem:PhaseEst} this
leads to a bounded-error algorithm with query complexity $O(A)$, and space complexity 
$\log(1+2nm) + O(\log A) = O(\log(nm))$, as claimed in \cref{thm:standard}.

\section{Robust Dual Adversary Algorithm}\label{sec:Robust}

In this section, we show that the dual adversary algorithm has robustness, 
in that it tolerates errors and flexibility in how it is defined. 
As described in \cref{sec:background}, we want to create a bounded-error quantum query algorithm for a Boolean function $f:X\rightarrow\{0,1\}$, for $X\subseteq\{0,1\}^n$.
Similar to the standard algorithm, 
our robust algorithm will involve applying phase estimation
to a unitary $U$ that acts on the space $\mathcal{H}=\mathbb{C}\oplus \mathbb{C}^n\otimes\mathbb{C}^m\otimes\mathbb{C}^2$.
We perform phase estimation on $U$ with a state $\ket{\hat{0}}\in \mathcal{H}$.

We now describe $U$. As in \cref{sec:background}, we define the orthogonal projector $\Pi_x = \proj{\hat{0}}+\sum_{i\in[n]}\proj{i}\otimes I\otimes\proj{x_i}$
on $\mathcal{H}$ where $I$ acts on $\mathbb{C}^m$, and $\ket{\hat{0}}$ is orthogonal to $\sum_{i\in[n]}\proj{i}\otimes I\otimes(\proj{0}+\proj{1})$. Notice that $\Pi_x$ can be implemented 
with two applications of the oracle $O_x$. Let $R$ be another unitary that acts on the same space as $\Pi_x$, but $R$ need not be a reflection. Let $U=(2\Pi_x-I)R$.

\begin{theorem}\label{thm:inexact_algorithm}
Let $\delta,\nu_x,\mu_x>0$, $\varepsilon_\psi,\varepsilon_\phi\geq 0$, and let $U=(2\Pi_x-I)R$ as defined above, so $U$ acts on $O(\log nm)$ qubits.
Consider $0 < \theta \leq 1$.
Suppose there are sets of (not necessarily normalized) vectors $\{\ket{\psi_x}=\frac{1}{\sqrt{\nu_x}}(\ket{\hat{0}}+\ket{\eta_x})\}_{x:f(x)=1}$ and $\{\ket{\phi_x}=\frac{1}{\sqrt{\mu_x}}(\ket{\hat{0}}+\ket{\eta_x})\}_{x:f(x)=0}$ where $\forall x\in X,\braket{\eta_x}{\hat{0}}=0$, and furthermore, that

\begin{enumerate}
\item $\forall x:f(x)=1, \quad \|(I-U)\ket{\psi_x}\|\leq \varepsilon_\psi$\label{part:psi_cond} and
\item $\forall x:f(x)=0, \quad \Pi_x\ket{\eta_x}=0$ and $\|(I+R)\ket{\phi_x}\|\leq \varepsilon_\phi$. \label{part:phi_cond}
\end{enumerate}

\noindent Then the probability of measuring a phase of $0$ if we do phase estimation on $U$ with initial state $\ket{\hat{0}}$  with precision $\theta$ and error $\delta$ when $f(x)=1$ is at least
\begin{equation}
 \left(\sqrt{\nu_x (1-\frac{5\varepsilon_{\psi}^2}{\theta^2})} - \sqrt{\nu_x\|\ket{\psi_x}\|^2-1} \right)^2(1-\delta)-\delta,
\end{equation}
and when $f(x)=0$ is at most
\begin{equation}
\mu_x\left(\varepsilon_\phi/2+\theta/2\|\ket{\phi_x}\|\right)^2+\delta.
\end{equation}
This algorithm uses $O\left(\frac{1}{\theta}\log\frac{1}{\delta}\right)$ queries and $O\left(\log(nm)+\log\frac{1}{\theta}\log\frac{1}{\delta}\right)$
qubits.
\end{theorem}

\cref{thm:inexact_algorithm} extends the robustness of the algorithm used to 
prove \cref{thm:standard} (as described in \cref{sec:background})
in several ways. In the standard analysis, $\varepsilon_\psi$ and $\varepsilon_\phi$
are both $0$, whereas we now allow them to be non-zero.
In addition,
\cref{thm:inexact_algorithm} allows for imperfect alignment between the vector sets $\{\ket{\psi_x}\}$ and $\{\ket{\phi_x}\}$ and the unitary $U$. This 
will be the key for our applications in the following sections.
Additionally, in the
standard algorithm, $R$ is chosen to be a reflection, but in 
\cref{thm:inexact_algorithm}, $R$ can be any unitary that satisfies the criterion 
of \cref{thm:inexact_algorithm}. While none of the applications we describe in this
paper use this flexibility in the design of $R$, it might be helpful in future use cases.

To prove \cref{thm:inexact_algorithm}, we will need the following two lemmas, which we prove
in \cref{sec:proofs_inexact}:
\begin{restatable}{lemma}{positiveStatePreserve}\label{lem:preserve}
Let $U$ be a unitary and $0 < \Theta \leq 1$. If $\|(I-U)\ket{\psi_x}\|^2\leq \varepsilon$, 
then $\|P_\Theta(U)\ket{\psi_x}\|^2\geq 1- \frac{1.1\varepsilon}{\Theta^2}$.
\end{restatable}

\cref{lem:preserve} tells us that if a unitary $U$ approximately preserves
 a state $\ket{\psi_x}$, then $\ket{\psi_x}$ has high overlap with the
low phase eigenspace of $U$. 
This  gives us flexibility when $f(x)=1$, in that
our initial state need not have high overlap with the $0$-phase space of $U$, but
instead we only require high overlap with the low-phase-eigenspace of $U$. 
The proof 
of \cref{lem:preserve} proceeds by 
decomposing $\ket{\psi_x}$ into its eigenbasis with respect to $U,$ 
and showing that $\epsilon$ serves to bound the amount of amplitude
$\ket{\psi_x}$ can have in states with eigenvalues larger than $\Theta$.

\begin{restatable}{lemma}{negativeStateDestroy}[Robust Approximate Spectral Gap Lemma]\label{lemma:approx_gap}
Let $\varepsilon\geq 0$, $\Pi$ be an orthogonal projector, $R$ be a unitary, and $U=(2\Pi-I)R$. For $\Theta> 0$, if $\|(I+R)\ket{w}\|\leq\varepsilon$, then
\begin{equation}
\|P_\Theta(U)\Pi\ket{w}\|\leq \frac{\varepsilon}{2}+\frac{\Theta}{2}\|\ket{w}\|.
\end{equation}
 \end{restatable} 

\cref{lemma:approx_gap} generalizes the standard 
approximate spectral gap lemma (\cref{lemma:approx_gap_standard}), in which $R$ is a reflection and $\varepsilon=0$.
In particular, \cref{lemma:approx_gap} shows that when $R$ is
not a reflection and when $\ket{w}$ is not exactly given a phase of $-1$ by $R$,
a variant of the approximate spectral gap lemma still holds.\footnote{We note \cref{lemma:approx_gap} is similar to 
Lemma 3.4 in \cite{jeffery2022multidimensional}, except we allow $R$ to not be a 
reflection.}
The proof of \cref{lemma:approx_gap} closely follows the proof approach of \cite[Lemma 4.2]{leeQuantumQueryComplexity2011},
which does not use Jordan's Lemma, but instead directly uses a series of triangle inequalities
and observations of preserved subspaces to obtain the result.

\begin{proof}[Proof of \cref{thm:inexact_algorithm}] 
We first consider the case of $x$ such that $f(x)=1$. By \cref{lem:PhaseEst}, the probability that we get an outcome
of $0$ when we perform phase estimation on the unitary $U$ with initial state $\ket{\hat{0}}$ with precision $\Theta$ and accuracy $\delta$ is at least $\|P_{\Theta/2}(U)\ket{\hat{0}}\|^2(1-\delta)-\delta$. Now
\begin{align}
\|P_{\Theta/2}(U)\ket{\hat{0}}\|&= \|P_{\Theta/2}(U) \sqrt{\nu_x}\ket{\psi_x} - P_{\Theta/2}(U) \ket{\eta_x} \| \nonumber\\
  &\geq \sqrt{\nu_x} \| P_{\Theta/2}(U) \ket{\psi_x} \| - \| P_{\Theta/2}(U) \ket{\eta_x} \| 
  \tag{reverse triangle inequality} \\
  &\geq \sqrt{\nu_x \left(1-\frac{5\varepsilon_\psi^2}{\Theta^2}\right)} - \|\ket{\eta_x} \|, 
\end{align}
where the first term in the final line combines \cref{lem:preserve} and the assumption that $\|(I-U)\ket{\psi_x}\|\leq \varepsilon_\psi$, so 
$\|(I-U)\ket{\psi_x}\|^2\leq \varepsilon_\psi^2$, and the second term uses the fact that projectors can only decrease the $\ell_2$ norm
of a vector. Thus, using that $\|\ket{\eta_x}\|=\sqrt{\nu_x\|\ket{\psi_x}\|^2-1}$, we have
\begin{align}
\|P_{\Theta/2}(U)\ket{\hat{0}}\|^2(1-\delta)-\delta\geq \left(\sqrt{\nu_x (1-\frac{5\varepsilon_{\psi}^2}{\Theta^2})} - \sqrt{\nu_x\|\ket{\psi_x}\|^2-1} \right)^2(1-\delta)-\delta.
\end{align}

When $f(x)=0$, by \cref{lem:PhaseEst}, the probability that we get an outcome
of $0$ when we perform phase estimation on $U$ with initial state $\ket{\hat{0}}$ with precision $\Theta$ and accuracy $\delta$ is at most
\begin{equation}\label{eq:thmx01}
\|P_\Theta(U)\ket{\hat{0}}\|^2+\delta=\mu_x\|P_\Theta(U)\Pi_x\ket{\phi_x}\|^2+\delta,
\end{equation}
since by assumption, $\Pi_x\ket{\phi_x}=1/\sqrt{\mu_x}\ket{\hat{0}}$.
Then from \cref{lemma:approx_gap} and our assumption that $\|(I+R)\ket{\phi_x}\|\leq \varepsilon_\phi$, we have
\begin{equation}\label{eq:thmx02}
\|P_\Theta(U)\Pi_x\ket{\phi_x}\|\leq \varepsilon_\phi/2+\Theta/2\|\ket{\phi_x}\|.
\end{equation}
Combining \cref{eq:thmx01} and \cref{eq:thmx02} gives us a probability of outcome $0$ of at most
\begin{equation}
\mu_x\left(\varepsilon_\phi/2+\Theta/2\|\ket{\phi_x}\|\right)^2+\delta.
\end{equation}

Finally the query complexity and space complexity come from the requirements
of phase estimation \cref{lem:PhaseEst}, and that $2\Pi_x-I$ 
can be implemented with two uses of the oracle.
\end{proof}

\section{Compressing the Dual Adversary Algorithm}\label{sec:JL}

In this section, we consider how and when it is possible to reduce
the space complexity of the quantum algorithm built from
the general adversary dual.
In particular, our goal is to take an
$f$-deciding vector set, reduce its dimension and hence create an algorithm
which requires fewer qubits to implement.

Our first result, \cref{thm:exact_compress}, is a simple compression scheme
that shows that an $f$-deciding vector set on an $n$-bit function with
maximum rank $\kap'$ and size $A$ can be compressed to an $f$-deciding vector
set with dimension $\kap'$ and size at most $A$ (see \cref
{def:decidingVecSet} for terminology). The number of qubits required by the
resulting algorithm is then $O(\log (n \kap'))$ and the query complexity is $O
(A)$, by \cref{thm:standard} and using that $A=O(n)$. Notice that $\kap'$ is at most the number of $1$-valued
inputs, but if many of the $f$-deciding vectors are linearly dependent, the
maximum rank could be much smaller.

We note that an $f$-deciding vector set is also an $\neg f$-deciding
vector set by \cref{def:decidingVecSet}, where $\neg f$ is the negation of
$f$. Also, a $\neg f$-deciding vector set can be used to design a bounded
error quantum algorithm for deciding $f$ by negating the output of the
algorithm. Thus for a given deciding vector set, we can minimize the space
used by the algorithm by considering either $\neg f$ or $f$.

\begin{restatable}{theorem}{exactCompression}\label{thm:exact_compress}
Given an $f$-deciding vector set with maximum rank $\kap'$ and size $A$,
 we can construct an $f$-deciding vector set with dimension $\kap'$ 
 and size at most $A$, resulting in an algorithm that decides $f$ with query complexity
 $O(A)$ and space complexity $O(\log(n\kap'))$.
\end{restatable}
\noindent To prove \cref{thm:exact_compress}, we apply a series of unitaries to rotate the vectors of the deciding
vector set into a smaller dimensional space. This is possible because if we
apply the same unitary to two vectors, their inner product is preserved. The full proof appears in \cref{sec:JLproofs}.

The next natural question is whether we can \textit{approximate} 
the solution to the general adversary dual in a lower dimension.
We answer this question by considering a compression of the 
vectors in a deciding vector set
with guarantees from the Johnson-Lindenstrauss lemma, which approximately
preserve the inner products of the vectors.

It turns out that this straightforward idea is not trivial to implement.
The first challenge is that we must preserve the tensor
product structure of our vectors in order to ensure that the query algorithm
can apply queries, so we must be careful about the part of the vectors
that we compress.
The second challenge is that the compression only approximately preserves 
the inner products of the compressed vectors so we need the 
robust dual adversary algorithm described in \cref{thm:inexact_algorithm}.

Formally stated in \cref{thm:jl_inexact}, given an $f$-deciding vector set with maximum rank $\kap'$ and size $A$, we show how to
build a quantum algorithm that succeeds with probability 2/3
and operates in a Hilbert space of dimension $O((\kap'^2 + A^4 \kap') n)$
with quantum query complexity $O(A)$.
Since $A=O(n)$, as discussed below \cref{thm:standard}, the number of qubits needed to run the algorithm is no more than
$O(\log (\kap' n))$.

Thus, this approximate compression using the
JL lemma achieves the same space complexity
as the exact compression of \cref{thm:exact_compress}, 
to within a constant multiplicative factor. 
This may seem surprising that we are not able to do better, 
since we are no longer requiring the constraints
are exactly satisfied. 
However, the compression dimension in the JL lemma has a polynomial dependence on the
allowed error, and 
since the amount of error we can tolerate roughly scales with maximum rank,
we do not get as much compression as one might hope for.

We now describe at a high level how we prove \cref{thm:jl_inexact}.
While an optimal deciding vector set might use complex vectors,
the following lemma, which we prove in \cref{sec:JLproofs},
shows that at a small cost in increased vector dimension, we can restrict to real vectors:
\begin{restatable}{lemma}{realOK}\label{lem:Complex_to_Real}
If there is an $f$-deciding vector set with complex numbers of dimension $m$ and size $A$, there exists an $f$-deciding vector set with only real numbers of dimension $2m$ and size $A$.
\end{restatable} 
The proof of \cref{lem:Complex_to_Real} proceeds by creating a new vector set from the
original where each new real vector consists of the real part of the
original complex vector stacked on top of the imaginary part of the original
complex vector.

While the standard Johnson-Lindenstrauss lemma guarantees that there is a
compression matrix that approximately preserves the $l_2$-norm difference of
any two vectors in a set, we use the following corollary, 
 which shows that the compression 
approximately preserves inner
products in addition to distances, which we prove in \cref{sec:JLproofs}. Our
proof of \cref{lem:JLcorr} is similar to a similar result in 
\cite{larsenOptimalityJohnsonLindenstraussLemma2017}, except that we do not assume
the vectors have norm $1$.

\begin{restatable}{corollary}{JLcorr}\label{lem:JLcorr}
Given $\varepsilon>0$, a set of finite vectors $V \subset \mathbb{R}^d$, and a number $N>8\ln(|V|)/\varepsilon^2$, 
there is a compression matrix $S \in \mathbb{R}^{N \times d}$ such that for $\ket{v},\ket{u}\in V$,
\begin{equation}
(S \ket{u})^\dagger (S \ket{v})=\braket{u}{v}\pm 2\varepsilon
(\| \ket{u}\|^2+\| \ket{v}\|^2).
\end{equation}
\end{restatable}

The Johnson-Lindenstrauss lemma guarantees the existence of such a compression
matrix $S$, and it can be found probabilistically by sampling random
projections. It requires $O(|X|n)$ time to find such a satisfying projection
via random sampling \cite{dasgupta2003elementary}. For our purposes, this
contributes to classical preprocessing time and space resources, and not
towards the quantum query complexity or quantum space use of the quantum algorithm itself. In particular, this
sampling requires no queries, so does not contribute to the query complexity.

We now describe how we use Johnson-Lindenstrauss to compress our deciding
vector set. Let $\{\ket{v_{x,i}}\}$ be a real $f$-deciding vector set with
size $A$, dimension $m$, and maximum rank $\kap'$. Define $\{\ket{\psi_x}\}_
{x:f(x)=1}$ and $\{\ket{\phi_x}\}_{x:f(x)=0}$ as in \cref
{eq:psidef,eq:phidef} in \cref{sec:background}.
Let $\kap$ be the rank of $\{\ket{\psi_x}\}_{x:f(x)=1}$.
(We show in \cref{lem:kappa_is_r} that $\kap' \leq \kap \leq 2n \kap'$.)

To compress our vectors, it suffices to compress their orthonormal basis. Let $\{\ket{\zeta_j}\}_{j\in[\kap]}$ be an orthonormal basis for the space spanned by $\{\ket{\psi_x}\}_{x:f(x)=1}$. 
Then there are (non-unique) real numbers $\{\alpha_{j,x}\in \mathbb{R}\}_{j\in[n],x\in f^{-1}(1)}$ such that
\begin{align}
\ket{\zeta_j} &= \sum_{x:f(x)=1} \alpha_{j,x} \ket{\psi_x} \label{eq:zetaj}.
\end{align} 
We will  approximately preserve the structure of the vectors $\ket{\zeta_j}$ in the compression.
Thus we define their components
\begin{align}
\ket{\SC_{j,i,b}}=\sum_{\substack{x:f(x)=1 \\ x_i = b}} \frac{\alpha_{j,x}}{\sqrt\nu_x}\ket{v_{x,i}},
\quad \forall j\in [\kap], i \in [n], b \in \{0,1\}
.\label{eq:zetajib}
\end{align}

We will use the random compression matrix $S$ from \cref{lem:JLcorr}
to compress the following set of vectors to error $\varepsilon$, as in 
\cref{lem:JLcorr}, (and the compression dimension $N$ will be chosen later to achieve the desired value of $\varepsilon$):
\begin{align}\label{eq:compression_vectors}
\left\{\ket{\SC_{j,i,b}}\right\}_{j\in[\kap],i\in [n],b\in\{0,1\}}
\bigcup \left\{\ket{v_{y,i}}\right\}_{y:f(y)=0,i\in [n]}.
\end{align}

We use these compressed vectors to define
\begin{align}\label{eq:compressedDefs}
\forall x:f(x)=1, \quad \ket{\psi_x'} &=  \left[\ketbra{\hat{0}}{\hat{0}} + (I \otimes S \otimes I)\right]\ket{{\psi_x}},\nonumber\\
\forall x:f(x)=0, \quad \ket{\phi_x'} &=  \left[\ketbra{\hat{0}}{\hat{0}} + (I \otimes S \otimes I)\right]\ket{{\phi_x}},\nonumber\\
\forall j\in[\kap],\qquad \ket{\zeta_j'} &= \left[\ketbra{\hat{0}}{\hat{0}} + (I \otimes S \otimes I)\right]\ket{\zeta_j}.
\end{align}
As one would expect, the primed, compressed versions of these vectors 
have approximately the properties of the uncompressed version, as follows:
\begin{restatable}{lemma}{compressedRelationsUpdate}
\label{lem:compressed_relations_update}
For $\{\ket{\zeta_j'}\}_{j\in[\kap]},$ $\{\ket{\psi_x'}\}_{x:f(x)=1}$ and $\{\ket{\phi_x'}\}_{x:f(x)=0}$ as described in \cref{eq:compressedDefs}, these vectors have the following properties
\begin{enumerate}
\item  $\forall j,l\in[\kap], \braket{\zeta_j'}{\zeta_l'}\in \delta_{j,l}\pm 4\varepsilon$ \label{part:zeta_preserve}
\item $\forall j\in [\kap],x\in f^{-1}(0),|\braket{\zeta_j'}{\phi_x'}|\leq2\varepsilon (\C+1) A$. \label{part:zeta_phi}
\item $\forall x:f(x)=1, \left|\|\ket{\psi_x'}\|^2-1\right|\leq 4\varepsilon \kap$ and
$\forall x:f(x)=0, \left|\|\ket{\phi_x'}\|-1\right|\leq 3\varepsilon$.\label{part:psiphiNormBounds}
\end{enumerate} 
\end{restatable}
\noindent The proof of \cref{lem:compressed_relations_update} uses \cref{lem:JLcorr} in fairly straightforward ways.

Let $\Delta'$ be the orthogonal projector onto the space spanned by $\{\ket{\psi_x'}\}_{x:f(x)=1}$ and the reflection $R$ be $2\Delta'-I$.
By definition, observe that 
$R\ket{\psi_x'}=\ket{\psi_x'}$, so \cref{thm:inexact_algorithm} \cref{part:psi_cond} is satisfied with $\varepsilon_\psi=0$.
Additionally, because $\ket{\phi_x'}$ has the structure
\begin{equation}
\ket{\phi_x'} \propto \ket{\hat{0}}+\sum_{i\in[n]}\ket{i}\ket{v_{x,i}'}\ket{\bar{x}_i}
\end{equation}
we have $\Pi_x\ket{\phi_x'}\propto\ket{\hat{0}}$, as required by \cref{thm:inexact_algorithm} \cref{part:phi_cond}. All that is left is to show that
$\|(I+R)\ket{\phi_x'}\|\leq \varepsilon_\phi$.

\begin{restatable}{lemma}{deltaPhi}\label{lem:deltaPhiAnalysis}
Consider $\varepsilon$ so that $\varepsilon \kap < 1/12$.
For $R$ as defined below \cref{lem:compressed_relations_update} and $\ket{\phi_x'}$ defined using \cref{eq:compressedDefs,eq:phidef}, we have
$\|(I+R)\ket{\phi_x'}\|\leq 8\varepsilon (\C+1) A \sqrt{\kap}$.
\end{restatable}
We prove \cref{lem:deltaPhiAnalysis} in \cref{sec:JLproofs}.
The main idea is to write $\Delta'$ in terms of the 
vectors $\{\ket{\zeta'}\}$, and then use \cref{lem:compressed_relations_update} \cref{part:zeta_phi}.
Along the way, we show how to use Gram-Schmidt to build
an orthonormal basis from $\{\ket{\zeta_j'}\}$, which is already almost 
orthonormal by \cref{lem:compressed_relations_update} \cref{part:zeta_preserve}
(see \cref{lemma:ortho_basis}, stated and proved in \cref{sec:JLproofs}).

With \cref{lem:deltaPhiAnalysis} in hand,
the conditions of \cref{thm:inexact_algorithm} are satisfied and 
we apply it to our compressed vectors.

\begin{theorem}\label{thm:jl_inexact}
  Consider a Boolean function $f:X\to\{0,1\}$ where $X \subseteq \{0,1\}^n$
  and an $f$-deciding vector set with maximum rank $\kap$ and size $A$.
  Using Johnson-Lindenstrauss compression,
  we can compress the $f$-deciding vector set
  to produce  a quantum algorithm that correctly evaluates $f(x)$ with probability $2/3$
  for every input $x\in X$ with $O(A)$ quantum query complexity.
  The algorithm uses $O(\log \kappa n)$ qubits.
\end{theorem}

\begin{proof}[Proof of \cref{thm:jl_inexact}]

We set 
\begin{align}\label{eq:set_theta_epsilon}
  \Theta=\frac{1}{4\sqrt{\C} A},
  \quad 
  \varepsilon = \min \left\{ \frac{1}{4 \C \kap}, 
  \frac{1}{32\sqrt{\C}(\C+1)A^2\sqrt{\kap}} \right\},
  \quad
  \delta = \frac{1}{25},
  \quad \text{and} \quad
  c = 100.
\end{align}
With these choices, we bound the failure probability below 1/3 for all inputs $x$. (Note $c$ appears in \cref{eq:psidef,eq:phidef}.)

\textbf{Case $\mathbf{f(x)=1}$:}
By \cref{thm:inexact_algorithm}, for $x$ such that $f(x)=1$, we have that when
$\varepsilon_\psi=0$, the probability of measuring a phase of $0$ is at least
\begin{align}\label{eq:final_success_1_prob}
\left(\sqrt{\nu_x } - \sqrt{\nu_x \|\ket{\psi_x'}\|^2-1} \right)^2(1-\delta)-\delta.
\end{align}

We know that $1 \leq \nu_x \leq 1+1/\C = 1.01$ from \cref{eq:nuxBound}
and $| \| \ket{\psi_x'} \|^2 - 1| \leq 4\varepsilon\kap \leq 1/100$
from \cref{lem:compressed_relations_update} \cref{part:psiphiNormBounds} and 
\cref{eq:set_theta_epsilon} when $\C=100$,
so \cref{eq:final_success_1_prob} is at least
\begin{equation}
\left(\sqrt{1}-\sqrt{1.01\cdot 1.01-1}\right)^2(1-\delta)-\delta\geq 2/3
\end{equation}
where the final inequality is satisfied for $\delta=1/25$.

\textbf{Case $\mathbf{f(x)=0}$:} 
By \cref{thm:inexact_algorithm}, for $x$ such that $f(x)=0$, we have that
the probability of measuring a phase of $0$ is at most
\begin{equation}
\mu_x\left(\varepsilon_\phi+\frac{\Theta}{2}\|\ket{\phi_x'}\|\right)^2+\delta.
\end{equation}
We know that $\mu_x\leq 1+\C A^2$ from \cref{eq:mu_x_bound},
$\varepsilon_\phi\leq 8 \varepsilon(\C+1)A\sqrt{\kap}$ 
from \cref{lem:deltaPhiAnalysis}, and
$\|\ket{\phi_x'}\|\leq 1+3\varepsilon$
from \cref{lem:compressed_relations_update} \cref{part:psiphiNormBounds}.
So the probability of measuring a $0$ phase is at most
\begin{equation}\label{eq:jl_0_failure}
(1+\C A^2)\left(8 \varepsilon(\C+1)A \sqrt{\kap}+\frac{\Theta}{2}(1+3\varepsilon)\right)^2+\delta.
\end{equation}

With $\Theta$, $\varepsilon$, $\delta$, and $\C$ as set in \cref{eq:set_theta_epsilon}, 
continuing from \cref{eq:jl_0_failure}, we have that probability of failure when $f(x)=0$ is at most
\begin{align}
(1+cA^2)\left(\frac{1}{4\sqrt{\C}A} + \frac{1}{4\sqrt{\C}A} \right)^2
  +\delta
  \leq \frac{1}{4\C A^2} + \frac{1}{4} + \delta \leq \frac{1}{3} 
\end{align}
since $A \geq 1$.

To achieve this compression with $\varepsilon$ as desired we look to \cref{lem:JLcorr} to see what compression dimension is achievable. From \cref{eq:compression_vectors}, we see we are compressing at most $3|X|n$
vectors. Thus we require a compression dimension
\begin{equation}
N=O(\log (|X|n)/\varepsilon^2)=O((\kappa^2+A^4\kappa)\log(|X|n)).
\end{equation}
Then since our unitary $U$ acts on $O(\log (nN))$ qubits, by \cref{thm:inexact_algorithm}, the space complexity is
\begin{equation}
O(\log nN+\log A)=O\left(\log\left(An(\kappa^2+A^4\kappa)\log(|X|n)\right)\right)=O(\log(\kappa n)).
\end{equation}
where we've used that $|X|\leq 2^n$ and $A=O(n).$ Also, from \cref{thm:inexact_algorithm}, the query complexity is $O(A)$.
\end{proof}

\begin{corollary}
If there are polynomially many $1$-valued or $0$-valued inputs to a function 
$f:X\rightarrow\{0,1\}$, then there is a query-optimal quantum algorithm that
evaluates $f$, and that uses $O(\log n)$ qubits.
\end{corollary}

\begin{proof}
Notice that $\kap \le n_1$ where $n_1$ the number of $1$-valued inputs to $f$.
When $\kap=O(n^d)$ for $d=O(1)$, by \cref{thm:exact_compress} or 
\cref{thm:jl_inexact}, we have that for any deciding vector set for $f$ with size $A$, we can create an algorithm with query complexity $O(A)$ whose space complexity is $O(\log n)$. 
Since there is always a deciding vector set for $f$ with size $A$ such that $O(A)$ is the optimal query complexity of $f$ \cite{reichardtReflectionsQuantumQuery2011}, our results imply that there exists
a query-optimal algorithm that uses logarithmic space. For the case of polynomially many $0$-valued inputs, we use $\neg f$.
\end{proof}

\section{Algorithm from a Numerical Solution}\label{sec:SVD}

While the dual adversary provides a method of designing
optimal query algorithms, in general it might be hard to find
an optimal solution. However, since the problem can be formulated as a
semidefinite program, we can find a numerical solution.
We show that numerical solutions that only approximately satisfy the dual adversary
constraints can be used to produce bounded-error quantum algorithms within
a constant factor of the objective function value of the numerical solution.

We first provide a theoretical result that shows how an algorithm can be 
constructed from an \textit{approximate} deciding vector set, a vector set that only approximately satisfies \cref{eq:filteredNorm}.
Then we introduce a classical Python package implemented specifically to solve 
the general adversary dual SDP and show that its solutions often satisfy the 
error bounds required by our theoretical results.

\subsection{Application of Robust Dual Adversary Algorithm}

In previous sections we assumed access to an exact solution to the general adversary dual
which we used to produce approximate solutions.
In this section, we use a finite precision SDP solver to obtain an approximate solution 
and then build a bounded-error quantum algorithm
with the robust dual adversary algorithm described in \cref{thm:inexact_algorithm}.

\begin{restatable}{theorem}{SvdAlg}\label{thm:svd_alg}
Consider a Boolean function $f: X \to \{0,1\}$
where $X \subseteq \{0,1\}^n$ 
and an approximate $f$-deciding vector set $\{\ket{v_{x,j}}\}_{x\in X,j\in [n]}$ in the sense that
\begin{align}\label{eq:numerical_error}
\varepsilon \coloneqq \max_{x,y: f(x) \neq f(y)} |\braket{\psi_x}{\phi_y}|,
\end{align}
is small as defined below, where $\ket{\psi_x},$ $\ket{\phi_x}$ are as in \cref{eq:psidef,eq:phidef}.
Let $A$ be the size and $m$ be the dimension of the 
approximate $f$-deciding vector set, with size and 
dimension defined as in \cref{def:decidingVecSet}.
Consider a matrix $M$ with rows $\ket{\psi_x}_{x:f(x)=1}$.
Let the singular values of $M$ be
$s_1 \geq s_2 \geq \cdots \geq s_{\kap} \geq s_{\kap + 1} \coloneqq 0$
and $n_1$ be $|\{x: f(x)=1\}|$.
If there exists $\kap^* \in [\kap]$ such that 
\begin{align}\label{eq:no_kap}
  \varepsilon \leq \frac{1}{\sqrt{n_1}} \left( \frac{s_{\kap^*}}{2\sqrt{\C} A} - s_{\kap^*+1} \right)
  \quad \text{and} \quad s_{\kap^*+1} \leq \frac{1}{2\sqrt{1000\C} A},
\end{align}
then there is a quantum algorithm that correctly
evaluates $f$ with probability at least $2/3$
with at most $O(A)$ queries and
$O(\log(nm))$ qubits.
\end{restatable}

Given any numerical solution, we can simply set $\kap^* = \kap$ in which case we require 
$\varepsilon < s_\kap/(2\sqrt{\C n_1}A)$. 
However, if the singular values fall off sharply then we can obtain a less stringent constraint on $\varepsilon$, which in turn leads
to less precision required by the numerical solver.
So in practice, we search for any $\kap^* \in [\kap]$ 
that gives a large enough bound to accommodate the $\varepsilon$ 
we observe in our numerical solution. 

To prove \cref{thm:svd_alg}, we apply \cref{thm:inexact_algorithm}.
For this 
application, we set $R$ to be equal to a reflection over the space
spanned by the first $\kappa^*$ right singular vectors of $M,$
the matrix whose rows are the vectors $\{\ket{\psi_x}\}_{x:f(x)=1}$.
We use the following two lemmas to show that this reflection
approximately preserves the vectors $\{\ket{\psi_x}\}_{x:f(x)=1}$
and mostly destroys vectors that are almost orthogonal
to all vectors in $\{\ket{\psi_x}\}_{x:f(x)=1}$.
Their proofs, found in \cref{sec:numericalProofs}, use standard results from the singular value decomposition approach to approximating
of matrices.

\begin{restatable}{lemma}{SvdPsi}\label{lemma:svd_psi}
Let $M$ be the matrix whose rows are the vectors $\{\ket{\psi_x}\}_{x\in X_1}$ for some
set $X_1$,
and denote $M$'s singular values by $s_1\geq s_2\geq\dots \ge s_{\kap}.$  Let $\Delta'$ be the orthogonal projector onto the first $\kap^*$
right singular vectors of $M$. Then $\forall x\in X_1$,
$(2\Delta'-I)\ket{\psi_x}=\ket{\psi_x}+\ket{\eta}$, where 
$\|\ket{\eta}\|\leq 2s_{\kap^*+1}$ and $s_{\kap^* + 1} = 0$ if $\kap^* = \kap$.
\end{restatable}

\begin{restatable}{lemma}{SvdPhi}\label{lemma:svd_phi}
For disjoint sets $X_1$ and $Y$ and  sets of vectors 
$\{\ket{\psi_x}\}_{x\in X_1}$ and $\{\ket{\phi_y}\}_{y\in Y}$
such that 
\begin{align}\label{eq:innerprod}
|\braket{\psi_x}{\phi_y}|\leq \varepsilon 
\quad \forall x \in X_1,y\in Y,
\end{align}
let $M$ be the matrix whose rows are the vectors $\{\ket{\psi_x}\}_{x\in X_1}$
and denote its singular values by $s_1\geq s_2\geq\dots \ge s_{\kap}$. Let $\Delta'$ be the orthogonal projector onto the first $\kap^*$
right singular vectors of $M$. 
Then $\forall y\in Y,$
$(2\Delta'-I)\ket{\phi_y}=\ket{\eta}$, where 
$\|\ket{\eta}\|\leq \frac{ s_{\kap^*+1}+\varepsilon\sqrt{|X_1|}}{s_{\kap^*}}$ and $s_{\kap^* + 1} = 0$ if $\kap^* = \kap$. 
\end{restatable}

We now use \cref{lemma:svd_psi,lemma:svd_phi} to apply \cref{thm:inexact_algorithm}
and prove \cref{thm:svd_alg}.

\begin{proof}[Proof of \cref{thm:svd_alg}]
Let $\Delta'$ be the orthogonal projector onto the first $\kap^*$
right singular vectors of $M$, and set $R=2\Delta'-I$ so as
in \cref{thm:inexact_algorithm}, $U =(2\Pi_x - I)R$.
Then
\begin{align}
\varepsilon_\psi \coloneqq \|(I-U)\ket{\psi_x}\| &= \|\ket{\psi_x} -(2\Pi_x - I)R \ket{\psi_x}\| \nonumber\\
&= \|\ket{\psi_x} - (2\Pi_x - I)(\ket{\psi_x} + \ket{\eta_x})\| \nonumber \\
&= \|(2\Pi_x - I)\ket{\eta_x}\|\leq 2 s_{\kap^*+1} \tag{By \cref{lemma:svd_psi}
and defs of $\Pi_x$, $\ket{\psi_x}$} \\
\varepsilon_\phi \coloneqq \|(I + R)\ket{\phi_y}\| &= \|2\Delta'\ket{\phi_y}\| 
\leq 2\frac{ s_{\kap^*+1}+\varepsilon\sqrt{n_1}}{s_{\kap^*}} \tag{By \cref{lemma:svd_phi}}.
\end{align}

Then by \cref{thm:inexact_algorithm}, if $f(x) = 1$, the probability of measuring a phase of 0 is at least
\begin{align}
 \left(\sqrt{\nu_x (1-\frac{5\varepsilon_{\psi}^2}{\Theta^2})} - \sqrt{\nu_x\|\ket{\psi_x}\|^2-1} \right)^2(1-\delta) - \delta &\ge 
\left(\sqrt{(1-\frac{20s_{\kap^*+1}^2}{\Theta^2})} - \sqrt{\nu_x\|\ket{\psi_x}\|^2-1} \right)^2(1-\delta)-\delta \\
&\ge 
\left(\sqrt{ 1-\frac{20s_{\kap^*+1}^2}{\Theta^2}} - \sqrt{1/c} \right)^2(1-\delta) -\delta\\
&\ge 2/3
\end{align}
where we set $c = 100$, $\delta = 1/25$, and $\Theta = (2\sqrt{\C} A)^{-1}$.
In addition, we used that
$\nu_x$ is between 1 and $1+1/c$ and that 
$\|\ket{\psi_x}\|=1$, both by \cref{eq:psidef}.
Recall that $s_{\kap^*+1}^2 \le \Theta^2/1000$ by assumption.

In the other case, where $f(x)=0$, the probability of measuring a phase of $0$ is at most
\begin{align}
\mu_x\left(2\frac{ s_{\kap^*+1}+\varepsilon\sqrt{n_1}}{s_{\kap^*}} +\Theta/2\|\ket{\phi_x}\|\right)^2+\delta & \le (1 + \C A^2)\left(2\frac{ s_{\kap^*+1}+\varepsilon\sqrt{n_1}}{s_{\kap^*}} +\Theta/2\right)^2+\delta \\
& = (1 + \C A^2)\left(\frac{1}{4\sqrt{\C}A}  +\frac{1}{4\sqrt{\C}A} \right)^2+\delta \\
& \le 1/3 \tag{by the proof of \cref{thm:jl_inexact}}
\end{align}
when %
\begin{align}
 \frac{ s_{\kap^*+1}+\varepsilon\sqrt{n_1}}{s_{\kap^*}} < \frac{1}{8\sqrt{\C}A} \iff \varepsilon < \frac{1}{\sqrt{n_1}}\left(\frac{s_{\kap^*}}{2\sqrt{\C}A} - s_{\kap^* + 1}\right).
\end{align}
The algorithm uses within a multiplicative factor of 
$\frac{1}{\Theta}\log(\frac{1}{\delta}) = 20A \log(20)$ queries
by \cref{lem:PhaseEst}.
\end{proof}

The obvious question is whether the conditions in \cref{thm:svd_alg}
are met in practice.
We show in the next subsection that the conditions are met in the vast majority
of numerical solutions to the general adversary dual for 
the random Boolean functions we find on functions of up to 
$25$ bits with domain size $32$.

\subsection{Experiments}

In this section, we describe how the numerical error
(defined in \cref{eq:numerical_error}) behaves in practice.

The general adversary dual (\cref{def:dual}) is a 
semidefinite programming (SDP) problem.
In order to solve the SDP numerically, 
we first reformulate it in the following standard form:

\begin{align}\label{eq:standard_form}
  \min_{\mathbf{X} \in \mathbf{S}^n} \tr(\mathbf{C X}) \quad \textrm{s.t.} 
  \quad \mathcal{A}(\mathbf{X}) = \mathbf{b} \quad \text{and} \quad \mathbf{X} \succeq 0
\end{align}
where $\mathbf{S}^n$ is the set of $n \times n$ symmetric matrices, 
$\mathbf{C} \in \mathbf{S}^n$, $\tr(\cdot)$ is the trace, 
$\mathbf{b} \in \mathbb{R}^m$ is the constraint vector, and
\begin{align}
  \mathcal{A}(\mathbf{X}) \coloneqq [\tr(\mathbf{A}^{1} \mathbf{X} ), \ldots, \tr(\mathbf{A}^{m} \mathbf{X}) ]^\top
\end{align}
for constraint matrices $\mathbf{A}^{i} \in \mathbf{S}^n$.
The reformulation requires converting the constraints
and introducing appropriate slack variables.
The constraints for our SDP
are particularly sparse which makes it inefficient
to use standard packages like CVXOPT and SDPA
\cite{diamond2016cvxpy, fujisawa2002sdpa}.
Instead, we use the alternating direction method of
\cite{wen2010alternating} which is specifically designed for
SDP problems with sparse structure and orthogonal constraints.
The pseudocode appears in \cref{alg:adm}.

\begin{algorithm}
  \caption{Alternating direction augmented Lagrangian method \cite{wen2010alternating}}\label{alg:adm}
  \begin{algorithmic}
  \Require{Constraint matrices $\mathbf{A, C }\in \mathbf{S}^n$, constraint vector $\mathbf{b} \in \mathbb{R}^m$, iterations $T$, tolerance $t\geq 0$}
  \Ensure{Output $\mathbf{X}$ approximately satisfies \cref{eq:standard_form}}
  \State $\mathbf{X}^0 \gets \mathbf{0}_{n \times n}$ \Comment{Zero matrix}
  \State $\mathbf{S}^0 \gets \mathbf{I}_{n \times n}$
  \For{$k=0,1,\ldots,T-1$}
  \State $y^{k+1} \gets -(\mathcal{A}\mathcal{A}^*)^+(\mathcal{A}(\mathbf{X}^k)-\mathbf{b} + \mathcal{A}(\mathbf{S}^k-\mathbf{C}))$
  \Comment{$(\cdot)^+$ denotes Moore-Penrose pseudoinverse}
  \State $\mathbf{V}^{k+1} \gets \mathbf{C}- \mathcal{A}^*(y^{k+1})-\mathbf{X}^k$
  \State $\mathbf{S}^{k+1} \gets \mathbf{Q}_+ \mathbf{\Sigma}_+ \mathbf{Q}_+^\top$ \Comment{$\mathbf{Q_+}, \mathbf{\Sigma_+}$ contain the non-negative eigendecomposition of $\mathbf{V}^{k+1}$}
  \State $\mathbf{X}^{k+1} \gets \mathbf{S}^{k+1} - \mathbf{V}^{k+1}$
  \State $\mathbf{X}^{k+1} \gets \operatorname{round}(\mathbf{X}^{k+1})$ \Comment{For entries within $t$ of 0 or $1$, round to nearest integer}
  \EndFor
  \State $\mathbf{X} \gets \mathbf{X}^{T-1}$
  \end{algorithmic}
\end{algorithm}

We slightly adapt \cref{alg:adm} to our problem:
We store each matrix in a sparse format, round entries of the solution
after every iteration, and use the Moore-Penrose pseudoinverse.

The first natural question is how well \cref{alg:adm} performs.
We test this by generating random Boolean functions
$f:X \to \{0,1\}$ where $X \subseteq \{0,1\}^n$ for different values of $n$.
Since the dimension of the SDP grows exponentially with $|X|$,
we fix $|X|=32$ for our experiments.
For each instance, we compute the maximum numerical error
\begin{align}
  \varepsilon \coloneqq \max_{x,y: f(x) \neq f(y)} |\braket{\psi_x}{\phi_y}|.
\end{align}
The results in \cref{thm:svd_alg} depend on the error $\varepsilon$ being small
so it's important that we can obtain solutions with small error efficiently.
\cref{fig:iterationsVsError} shows how $\varepsilon$ 
decreases with the number of iterations $T$ 
of the SDP solver.

The next question is whether the error is sufficiently small
to satisfy the requirements of \cref{thm:svd_alg}.
\cref{fig:allowed_v_actual} shows how $\varepsilon$ is typically much smaller
than the bound required in \cref{thm:svd_alg} for random 
functions of up to $25$ bits with domain size $32$.
Our experiments test for whether $\varepsilon$ satisfies the bound of \cref{thm:svd_alg} in the case that $\kap^* = \kap$.
It could be that even larger $\varepsilon$ is tolerated by considering $\kap^* < \kap$.

\begin{figure}[H]
    \centering
    \includegraphics[scale=.9]{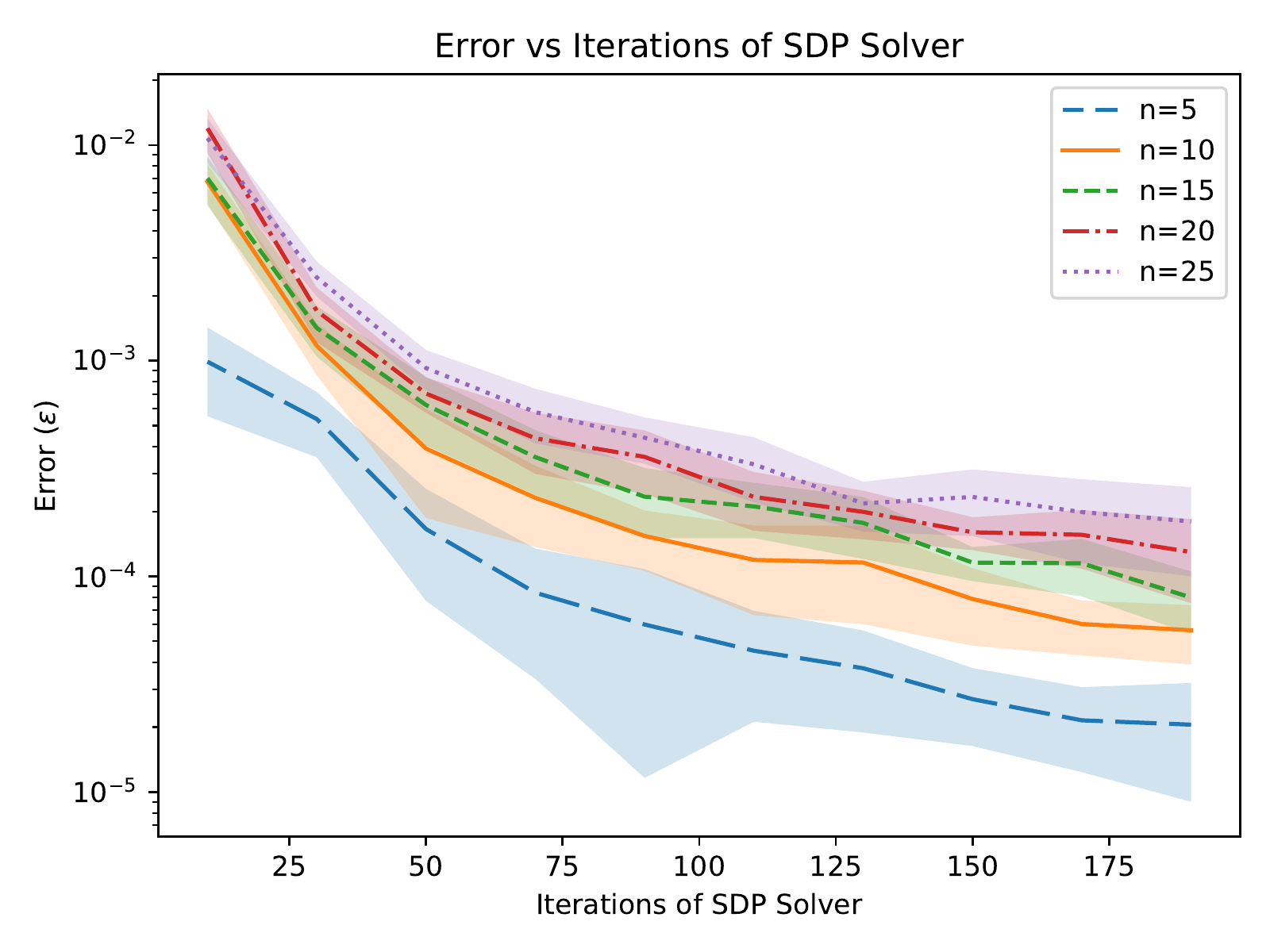}
    \captionsetup{width=.9\linewidth}
    \caption{For random Boolean functions with
    a domain of fixed size, the maximum numerical error 
    decreases with the number of iterations of the SDP solver.
    Note: the vertical axis is on a logarithmic scale
    and the shaded regions contain one standard deviation from
    20 random instances.}
    \label{fig:iterationsVsError}
\end{figure}

\begin{figure}[H]
    \centering
    \captionsetup{width=.9\linewidth}
    \includegraphics[scale=.9]{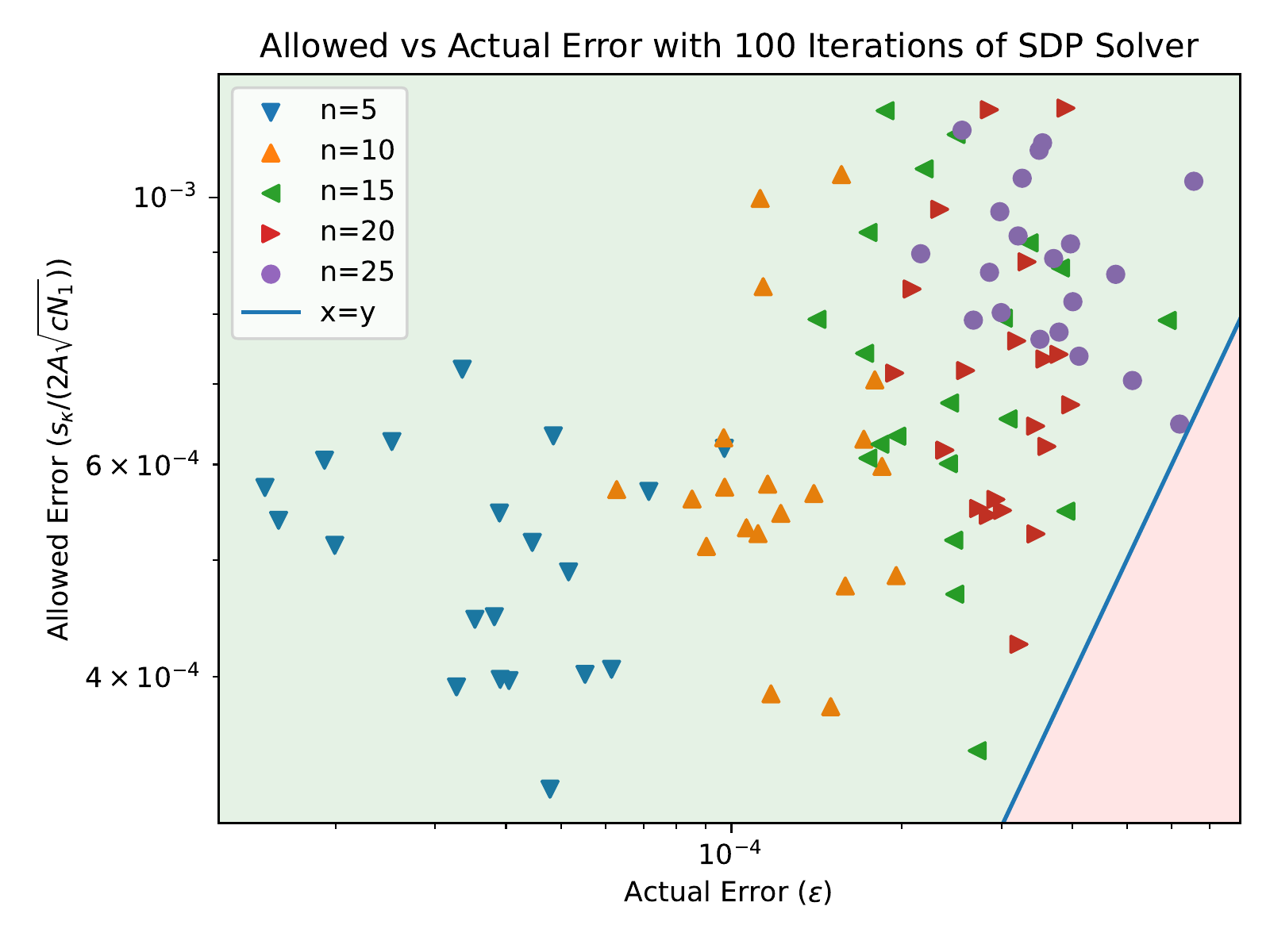}
    \caption{For random Boolean functions with a domain of fixed size,
    the maximum numerical error 
    stays below the threshold required to construct a provably correct
    bounded-error quantum algorithm.
    Note: both axes are on a logarithmic scale and the green region
    above the diagonal line indicates the error is small enough.}
    \label{fig:allowed_v_actual}
\end{figure}

\section{Acknowledgments}
We thank Stacey Jeffery for elucidating discussions.
SK was sponsored by the U.S. Army Research Office and this work was accomplished
under Grant Number W911NF-20-1-0327. The views and conclusions contained in
this document are those of the authors and should not be interpreted as
representing the official policies, either expressed or implied, of the Army
Research Office or the U.S. Government. The U.S. Government is authorized to
reproduce and distribute reprints for Government purposes notwithstanding
any copyright notation herein. 
RTW was supported by NSF No. 1922658 and 1916505.

\bibliography{spanProgram}
\bibliographystyle{alpha}

\appendix

\section{Phase Estimation Proof}\label{sec:prelimProofs}

In this section, we prove the performance guarantees on the Phase Estimation circuit that we use in our algorithms.
\phaseEst*

\begin{proof} 
Let $\overline{P}_\Theta(U)=I-P_\Theta$(i.e. $\overline
{P}_\Theta$ is the orthogonal projector onto all eigenvectors with phase
more than $\Theta$).
  
We use the parallel phase estimation plus median voting approach of
Ref.~\cite{nagaj2009fast}, which builds off of Refs.~\cite
{kitaevQuantumMeasurementsAbelian1995,cleveQuantumAlgorithmsRevisited1998}.
For an eigenstate $\ket{\lambda}$ of $U$ with phase $\phi$, with
probability at least $(1-\delta)$, this protocol outputs an estimate that
is within $\Theta/2$ of $\phi$ and uses a single preparation of $\ket
{\lambda}$, an additional $O\left(\log\frac{1}{\Theta}\log\frac{1}
{\delta}\right)$ workspace qubits to store the parallel phase estimates,
and $O\left(\frac{1}{\Theta}\log\frac{1}{\delta}\right)$ controlled
applications of $U$. Let $D(U)$ be the phase estimation circuit on $U$,
which acts on $n+g$ qubits, for $g=O\left(\log\frac{1}{\Theta}\log\frac{1}
{\delta}\right).$ Thus for an input state $\ket{\psi}$ on $n$ qubits, we
append $\ket{0}^{\otimes g}$ in order to apply the circuit to it.  Let
$\Gamma_0$ be the projector onto the space that corresponds to a $0$ output.
 Then for
an input state $\ket{\psi}$, the probability of outputting a phase of $0$
after the phase estimation circuit is $\|\Gamma_0 D(U)\left(\ket{\psi}\ket
{0}^{\otimes g}\right)\|^2$.

For the upper bound, using the triangle inequality, we have
\begin{align}\label{eq:upperBound1}
\|\Gamma_0 D(U)\left(\ket{\psi}\ket{0}^{\otimes g}\right)\|^2\leq& \left(\|\Gamma_0 D(U)\left((P_\theta\ket{\psi})\ket{0}^{\otimes g}\right)\|
+\|\Gamma_0 D(U)\left((\overline{P}_\theta\ket{\psi})\ket{0}^{\otimes g}\right)\|\right)^2\nonumber\\
\leq & \left(\|\Gamma_0 D(U)\left((P_\theta\ket{\psi})\ket{0}^{\otimes g}\right)\|+\sqrt{\delta}\right)^2,
\end{align} 
where we have used that $\|\Gamma_0 D(U)\left((\overline
{P}_\theta\ket{\psi})\ket{0}^{\otimes g}\right)\|^2\leq \delta$, since
$\overline{P}_\theta\ket{\psi}$ is a linear combination of eigenvectors with
eigenvalues more than $\theta/2$ away from $0$. Continuing from \cref
{eq:upperBound1}, and using that $\|\Gamma_0 D(U)\left((P_\theta\ket
{\psi})\ket{0}^{\otimes g}\right)\|\leq \|P_\theta\ket{\psi}\|\leq 1$, we
have
\begin{equation}\label{eq:upperBound2}
\|\Gamma_0 D(U)\left(\ket{\psi}\ket{0}^{\otimes g}\right)\|^2\leq\|P_\theta\ket{\psi}\|^2+2\sqrt{\delta}+\delta\leq \|P_\theta\ket{\psi}\|^2+3\sqrt{\delta}.
\end{equation} 

For the lower bound, using the reverse triangle inequality, we have
\begin{align}\label{eq:lowerBound1}
\|\Gamma_0 D(U)\left(\ket{\psi}\ket{0}^{\otimes g}\right)\|^2\geq& 
\left(\|\Gamma_0 D(U)\left((P_{\theta/2}\ket{\psi})\ket{0}^{\otimes g}\right)\|
-\|\Gamma_0 D(U)\left((\overline{P}_{\theta/2}\ket{\psi})\ket{0}^{\otimes g}\right)\|\right)^2\nonumber\\
\geq & \left(\sqrt{1-\delta}\|(P_{\theta/2}\ket{\psi}\|-\sqrt{\delta}\right)^2,
\end{align} 
where we have used that $\|\Gamma_0 D(U)\left((P_{\theta/2}\ket
{\psi})\ket{0}^{\otimes g}\right)\|^2\geq (1-\delta)\|(P_{\theta/2}\ket
{\psi}\|^2$ since $P_{\theta/2}\ket{\psi}$ consists of eigenvectors of $U$
that are within $\theta/2$ of phase $0$, (and not within $\theta/2$ of the
next possible phase output, which is $\theta$). Likewise, 
$\|\Gamma_0 D(U)\left((\overline{P}_{\theta/2}\ket{\psi})\ket{0}^
{\otimes g}\right)\|^2\leq \delta$ because $\overline{P}_{\theta/2}\ket
{\psi}$ is a linear combination of eigenvectors with eigenvalues more than
$\theta/2$ away from $0$. Continuing from \cref{eq:lowerBound1}, and using
that $\|\Gamma_0 D(U)\left((P_{\theta/2}\ket{\psi})\ket{0}^
{\otimes g}\right)\|\leq \|P_\theta\ket{\psi}\|\leq 1$, we have
\begin{equation}\label{eq:lowerBound2}
\|\Gamma_0 D(U)\left(\ket{\psi}\ket{0}^{\otimes g}\right)\|^2\geq 
(1-\delta)\|(P_{\theta/2}\ket{\psi}\|^2-2\sqrt{\delta}.
\end{equation}

Because $\delta$ appears in the space and query complexity asymptotically as
$O(\log(1/\delta)$, we can replace the $O(\sqrt{\delta})$ terms in our bounds
with $\delta$ without affecting the result.
\end{proof}

\section{Proofs for the Robust Dual Adversary Algorithm}\label{sec:proofs_inexact}

In this section, we prove \cref{lem:preserve} and \cref{lemma:approx_gap},
which are used to prove \cref{thm:inexact_algorithm} in \cref{sec:Robust}.

\positiveStatePreserve*

\begin{proof}[Proof of \cref{lem:preserve}]
We can decompose $\ket{\psi_x}$ into its eigenbasis with respect to $U$, so 
\begin{equation}
\ket{\psi_x}=\sum_ic_i\ket{\lambda_i}
\end{equation}
where $\ket{\lambda_i}$ is an eigenvector of $U$ with eigenphase $\beta_i$,
and $c_i$ is the amplitude of $\ket{\lambda_i}$.
By rearranging terms, we can write
\begin{equation}
U\ket{\psi_x}= \ket{\psi_x} + U \ket{\psi_x}  - \ket{\psi_x} =
\ket{\psi_x}+\sum_{i}c_i(e^{i \beta_i}-1)\ket{\lambda_i}.
\end{equation}
By assumption, we have
\begin{equation}
\varepsilon\geq \left\|\sum_{i}c_i(e^{i \beta_i}-1)\ket{\lambda_i}\right\|^2.
\end{equation}

Since $\{\ket{\lambda_i}\}$ form an orthonormal basis, we have 
\begin{equation}
\varepsilon\geq \sum_{i}|c_i|^2\left|(e^{i \beta_i}-1)\right|^2
\geq \sum_{i:|\beta_i|>\Theta}|c_i|^2\left|(e^{i \beta_i}-1)\right|^2
\geq \left|(e^{i \Theta}-1)\right|^2
\sum_{i:|\beta_i|>\Theta}|c_i|^2
.
\end{equation}
where, for the second inequality, we have used that all the terms are positive
and, for the final inequality, we used that $|e^{i a}-1|^2 \geq |e^{i b}-1|^2$
for $\pi \ge |a| \geq |b| \ge 0$. 

Finally, using the inequality $|e^{i\Theta}-1|^2\geq \Theta^2/1.1$, which holds for $\Theta \leq 1$, 
we have
\begin{equation}
\varepsilon\geq \frac{\Theta^2}{1.1}\sum_{i:|\beta_i| > \Theta}|c_i|^2 
= \frac{\Theta^2}{1.1}(1-\|P_{\Theta}(U)\ket{\psi_x}\|^2),
\end{equation}
where in the last equality, we've used the definition of $P_\Theta(U).$

Rearranging, we have
\begin{equation}
\|P_\Theta(U)\ket{\psi_x}\|^2\geq 1-\frac{1.1\varepsilon}{\Theta^2}.
\end{equation}
\end{proof}

\negativeStateDestroy*
\begin{proof}[Proof of \cref{lemma:approx_gap}] (We follow the approach
of \cite{leeQuantumQueryComplexity2011}, which does not use Jordan's Lemma.)
For ease of notation, we abbreviate $P_\Theta(U)$ as $P_\Theta$. Let 
$\{\ket{\lambda_i}\}$ denote the eigenvectors of $U$,
$\beta_i$ denote the phase of $\ket{\lambda_i}$, 
and $\bar{\Pi}$ denote $I-\Pi$.
Then define $\ket{v}=P_\Theta\Pi\ket{w}$, 
$\ket{v'}=R\ket{v}$, 
and $\ket{v''}=(2\Pi-I)\ket{v'}=U\ket{v}$.
Now
\begin{align}\label{eq:approxSP}
 \|\ket{v}-\ket{v''}\|^2&= \|\ket{v}-U\ket{v}\|^2
 \nonumber\\
 &= \|(I-U)P_\Theta\ket{v}\|^2 &\textrm{ since }\ket{v}=P_\Theta\Pi\ket{w} \textrm{ and }P_\Theta P_\Theta = P_\Theta
 \nonumber\\
 &= \left\|\sum_{i:|\beta_i|\leq\Theta}\left(1-e^{i\beta_i}\right)\proj{\lambda_i}\ket{v}\right\|^2
 \nonumber\\
 &\leq 2(1-\cos(\Theta))\left\|\sum_{\ket{\lambda_i}:|\beta_i|\leq\Theta}\proj{\lambda_i}\ket{v}\right\|^2
 \nonumber\\
 &\leq 2(1-\cos(\Theta))\left\|\ket{v}\right\|^2
 \nonumber\\
 &\leq\Theta^2\|\ket{v}\|^2.
\end{align}

Also, since $\ket{v}+\ket{v'}=\ket{v}+R\ket{v}$,
we have
\begin{equation}\label{eq:approxSP1}
(\bra{v}+\bra{v'})\ket{w}=\bra{v}(I+R)\ket{w}\leq \| \ket{v} \| \varepsilon
\end{equation}
by Cauchy-Schwarz and our assumption that $\|(I+R)\ket{w}\|\leq \varepsilon$.

Now $\ket{v'}+\ket{v''}=\ket{v'}+(2\Pi - I)\ket{v'}=2\Pi\ket{v'}$.
So $\Pi(\ket{v'}+\ket{v''})=
\ket{v'}+\ket{v''}$
where we used that $\Pi$ is an orthogonal projector. 
Also
 $\ket{v'}-\ket{v''}=\ket{v'}-(2\Pi-I)\ket{v'}=2(I-\Pi)\ket{v'}$, so $\bar{\Pi}(\ket{v'}-\ket{v''})=\ket{v'}-\ket{v''}$ because $\bar{\Pi}=I-\Pi$ is an orthogonal projector. This last statement implies $\Pi\ket{v'}=\Pi\ket{v''}$, and the previous implies $\bar{\Pi}\ket{v'}=-\bar{\Pi}\ket{v''}$. So
\begin{align}\label{eq:approxSP2}
(\bra{v}+\bra{v'})\ket{w}&=(\bra{v}+\bra{v'})(\Pi+\bar{\Pi})\ket{w}
\nonumber\\
&=(\bra{v}+\bra{v'})\Pi\ket{w}+(\bra{v}+\bra{v'})\bar{\Pi}\ket{w}
\nonumber\\
&=(\bra{v}+\bra{v''})\Pi\ket{w}+(\bra{v}-\bra{v''})\bar{\Pi}\ket{w}.
\end{align} 

We can now prove the bound of interest:
\begin{align}\label{eq:bigLongBound}
\|P_\Theta \Pi \ket{w} \|^2 &
= \braket{v}{v}
=\bra{w}\Pi P_\Theta P_\Theta \Pi\ket{w}
=\left(\bra{w}\Pi P_\Theta \right)\Pi\ket{w}
=\bra{v}\Pi\ket{w} \nonumber \\
&=\frac{1}{2}\left((\bra{v}+\bra{v''})\Pi\ket{w}+(\bra{v}-\bra{v''})\Pi\ket{w}\right)
\nonumber\\
&=\frac{1}{2}\left((\bra{v}+\bra{v'})\ket{w}-(\bra{v}-\bra{v''})\bar{\Pi}\ket{w}+(\bra{v}-\bra{v''})\Pi\ket{w}\right)
\tag{using \cref{eq:approxSP2}}\nonumber\\
&\leq\frac{1}{2}\left(\|\ket{v}\|\varepsilon+\left\|\ket{v}-\ket{v''}\right\|\cdot\left\|(\Pi-\bar{\Pi})\ket{w}\right\|\right), \tag{Using \cref{eq:approxSP1} and Cauchy-Schwarz}
\\
&\leq \frac{1}{2}\left(\|\ket{v}\|\varepsilon+\Theta^2\|\ket{v}\|\left\|\ket{w}\right\|\right),
\end{align}
where in the last line we have used \cref{eq:approxSP}. Dividing both sides by $\|\ket{v}\|=\|P_\Theta \Pi \ket{w} \|$,
we have
\begin{align}
\|P_\Theta \Pi \ket{w} \|&\leq\frac{\varepsilon}{2}+\frac{\Theta}{2}\|\ket{w}\|.
\end{align}

\end{proof}

\section{Proofs of Space Compression}\label{sec:JLproofs}
In this appendix, we include proofs related to compressing the space
used by a dual adversary algorithm.

We first prove our result for exact compression, which rotates
the deciding vector set into a smaller dimensional space:
\exactCompression*

\begin{proof}[Proof of \cref{thm:exact_compress}]
Let $\{\ket{v_{x,j}}\}$ be an $f$-deciding vector set with maximum rank $\kap$. Then by \cref{def:decidingVecSet},
\begin{align}
\kap = %
\max_{j \in [n]} 
\operatorname{rank}\{\ket{v_{x,j}} : f(x) = 1\}.
\end{align}
Let $\kap_j$ be the rank of the set $\{\ket{v_{x,j}} : f(x)=1 \}$.
For each $j \in [n]$, let $\{\ket{\gamma_{ji}}\}_{i \in [\kap_j]}$ 
form an orthonormal basis of the space.
We next define a compression matrix 
$S_j \in \mathbb{C}^{m \times m}$, which rotates the
space into the first $\kap_j$ standard basis states, 
removing components outside this space.
In particular, set
\begin{align}
  S_j &\coloneqq \sum_{i=1}^{\kap_j} \ketbra{i}{\gamma_{ji}}. 
\end{align}
Then for all $x\in X$, set $\ket{v_{x,j}'}=S_j\ket{v_{x,j}}$.
We will show that $\{\ket{v_{x,j}'}\}$ is an $f$-deciding vector with dimension $\kap$ and size at most $A$.

First, note that $\forall j\in[n],x\in f^{-1}(1), \ket{v_{x,j}'}\in \mathbb{C}^\kappa$, so all vectors are of dimension $\kappa$. Next, to show that $\{\ket{v_{x,j}'}\}$ is an $f$-deciding vector with dimension $\kap$ and size at most $A$, consider a unitary $V_j\in\mathbb{C}^m\times \mathbb{C}^m$ that acts as $S_j$ on $\{\ket{\gamma_{ji}}\}_{i \in [\kap_j]}$,
and acts on orthogonal states in any way needed to make $V_j$ unitary. 
Then we can rewrite $S_j$ as 
\begin{equation}
S_j=\Pi_{\kap_j} V_j,
\end{equation}
where 
\begin{equation}
\Pi_{\kap_j}=\sum_{i=1}^{\kap_j}\proj{i}.
\end{equation}
That is, $\Pi_{\kap_j}$ is the projection onto the first $\kap_j$ standard basis states.
Also note that for $\forall x\in f^{-1}(1)$,  
\begin{equation}\label{eq:project_effect}
S_j\ket{v_{x,j}}=\Pi_{\kap_j}V_j\ket{v_{x,j}}=V_j\ket{v_{x,j}}.
\end{equation}

To show that $\{\ket{v_{x,j}'}\}$ is an $f$-deciding vector
set, it suffices to show that for any $j \in [n]$ and $x\in f^{-1}(1),y \in f^{-1}(0)$, $\braket{v_{x,j}'}{v_{y,j}'} = \braket{v_{x,j}}{v_{y,j}}$. We have
\begin{align}
  \braket{v_{x,j}'}{v_{y,j}'} &= \bra{v_{x,j}} S_j^\dagger S_j \ket{v_{y,j}} \nonumber\\
  &= \bra{v_{x,j}} V_j^\dagger \Pi_{r_j}^\dagger \Pi_{r_j} V_j \ket{v_{y,j}} \nonumber \\
    &= \bra{v_{x,j}} V_j^\dagger  V_j \ket{v_{y,j}} \nonumber\\
  &= \braket{v_{x,j}}{v_{y,j}}.
\end{align}
where in the third line we have used \cref{eq:project_effect}.

To show that the size of $\{\ket{v_{x,j}'}\}$ is at most $A$, we note that our compression does not increase the norm of our vectors. For any $j\in [n],x\in X$,
\begin{align}
\left\|\ket{v_{x,j}'}\right\|=\left\|\Pi_{r_j} V_j\ket{v_{x,j}}\right\|\leq 
\left\| V_j\ket{v_{x,j}}\right\|=\left\|\ket{v_{x,j}}\right\|,
\end{align}
where we have used that projectors can only decrease the norm of a vector,
and unitaries can not change the norm.
Now the vectors $\{\ket{v_{x,j}'}\}_{j\in[n],x\in X}$ have $0$ amplitude
on all standard basis states larger than $\kap$, so we can truncate
them to be vector in $\mathbb{C}^\kappa$
so we have found an $f$-deciding vector set with dimension $\kappa$
and size at most $A$.
\end{proof}

The following results are related to our application of JL
compression to deciding vector sets.
Because Johnson-Lindenstrauss compression theorems generally apply to real 
vectors, we The following lemma shows that we can restrict to real vectors in 
our deciding vector set without affecting the size, and at a cost of a factor
of $2$ in the dimension:
\realOK*

\begin{proof}[Proof of \cref{lem:Complex_to_Real}]
Let 
\begin{equation}
\ket{v_{x,j}^\mathbb{R}}=\textrm{Re}\left(\ket{v_{x,j}}\right)\oplus \textrm{im}\left(\ket{v_{x,j}}\right)
\end{equation}
where $\textrm{Re}(\cdot)$ and $\textrm{im}(\cdot)$ are the real and imaginary parts of a complex number, respectively.
Then for $x,y\in X:f(x)\neq f(y)$,
\begin{equation}
\sum_{j:x_j\neq y_j}\braket{v_{x,j}^\mathbb{R}}{v_{y,j}^\mathbb{R}}
=
\sum_{j:x_j\neq y_j}\textrm{Re}(\braket{v_{x,j}}{v_{y,j}})
=1.
\end{equation}
If $\ket{v_{x,j}}\in \mathbb{C}^m$, then $\ket{v_{x,j}^\mathbb{R}}\in\mathbb{R}^{2m}.$
\end{proof}

The following is the standard statement of the JL Lemma:
\begin{lemma}[Johnson-Lindenstrauss Lemma \cite{johnson1984extensions, duchi2016lecture}]
Given $\varepsilon>0$, a set of finite vectors $V \subset \mathbb{R}^d$, and a number $N>8\ln(|V|)/\varepsilon^2$, 
there is a compression matrix $S \in \mathbb{R}^{N \times d}$ such that for $\ket{v},\ket{u}\in V$,
\begin{equation}\label{eq:JL}
(1-\varepsilon)\|\ket{u}-\ket{v}\|^2\leq \|S \ket{u}- S\ket{v}\|^2\leq (1+\varepsilon)\|\ket{u}-\ket{v}\|^2.
\end{equation}
\label{lem:JL}
\end{lemma}

For our application, we need to show that JL compression also approximately 
preserves the inner product of the compressed vectors. We show this in the 
following corollary, which is similar to that in \cite{larsenOptimalityJohnsonLindenstraussLemma2017}, except that we do not assume the vectors have norm $1$.
\JLcorr*

\begin{proof}
Let $V'$ be $V\cup \ket{\vec{0}}$,
where $\ket{\vec{0}}\in \mathbb{R}^{d}$ is the all $0$'s vector. 
By \cref{lem:JL}, there is a compression matrix $S$ that maps each $\ket{v}\in V'$
to $ S \ket{\vec{0}}\in \mathbb{R}^{N}$ such that \cref{eq:JL} is satisfied.
By adding the origin $\ket{\vec{0}}$, we will ensure that the mapping $S$ preserves vector norms.
In particular, since $S$ is a linear mapping, $S \ket{\vec{0}} =\ket{\vec{0}'}$, where $\ket{\vec{0}'}$ is the all $0$'s vector on $\mathbb{R}^{N}$.
Then using \cref{eq:JL}, we have that for each $\ket{v}\in V$,
\begin{equation}
(1-\varepsilon)\|\ket{u}-\ket{\vec{0}}\|^2\leq \| S \ket{u}- S\ket{\vec{0}}\|^2\leq (1+\varepsilon)\|\ket{u}-\ket{\vec{0}}\|^2,
\end{equation}
so 
\begin{equation}\label{eq:normPreserve}
(1-\varepsilon)\|\ket{u}\|^2\leq \|S \ket{u} \|^2\leq (1+\varepsilon)\|\ket{u}\|^2,
\end{equation}
and we see that the mapping $S$ approximately preserves the norms of the vectors
in our set.

Because the vectors $\{ S \ket{v} \}_{\ket{v}\in V}$ are real, 
using the inner product expression for the squared $\ell_2$-norm difference, 
we have that for any $\ket{u},\ket{v}\in V$,
\begin{align}
  (S \ket{u})^\dagger (S \ket{v})&=
  \frac{1}{2} \left(
  \| S \ket{u} \|^2
  + \| S \ket{v} \|^2
  - \| S \ket{u} - S \ket{v} \|^2
  \right) 
  \nonumber\\
  & \leq \frac{1+\varepsilon}{2}
  (\| \ket{u} \|^2 + \| \ket{v} \|^2)
  - \frac{1-\varepsilon}{2} \| \ket{u} - \ket{v} \|^2 
  \nonumber\\&\leq \braket{u}{v} + \varepsilon
  (\| \ket{u} \|^2 + \| \ket{v} \|^2 + |\braket{u}{v}|)
  \nonumber\\ &= \braket{u}{v} + \varepsilon
  \left(\| \ket{u} \|^2 + \| \ket{v} \|^2 + \sqrt{\|\ket{u}\|^2\|\ket{v}\|^2}\right) \tag{by Cauchy-Schwarz}
  \nonumber\\ &= \braket{u}{v} + \varepsilon
  \left(\| \ket{u} \|^2 + \| \ket{v} \|^2 + \frac{\| \ket{u} \|^2 + \| \ket{v} \|^2}{2}\right) \tag{geometric vs. arithmetic mean}
  \nonumber\\& \leq \braket{u}{v} +\frac{3\varepsilon}{2} (\|\ket{u} \|^2 + \| \ket{v} \|^2) .
\end{align}
where the first inequality uses \cref{eq:normPreserve,eq:JL}.
The other direction follows similarly.
\end{proof}

We next use \cref{lem:JLcorr} to show that the compressed vectors in our
JL compression of the dual adversary algorithm have approximately the same size and inner products as their corresponding uncompressed vectors:
\compressedRelationsUpdate*
\begin{proof}[Proof of \cref{lem:compressed_relations_update} \cref{part:zeta_preserve}.]

We express $\ket{\zeta_j'}$ and $\ket{\zeta_j}$ in terms of our compression
vectors $C$ (see \cref{eq:compression_vectors}):
\begin{align}\label{eq:zetap_in_zeta}
\ket{\zeta_j'}=&\sum_{x:f(x)=1}\alpha_{j,x}\ket{\psi_x'}\nonumber\\
=&\sum_{x:f(x)=1}\frac{\alpha_{j,x}}{\sqrt{\nu_x}}\left(\ket{\hat{0}}+\frac{1}{\sqrt{\C A}}\sum_{i\in[n]}\ket{i}\ket{v_{x,i}'}\ket{x_i}\right)\nonumber\\
=&\left(\sum_{x:f(x)=1}\frac{\alpha_{j,x}}{\sqrt{\nu_x}}\ket{\hat{0}}\right)+\frac{1}{\sqrt{\C A}}\sum_{i\in[n], b\in\{0,1\}}\ket{i}\left(\sum_{x:f(x)=1,x_i=b}\frac{\alpha_{j,x}}{\sqrt{\nu_x}}\ket{v_{x,i}'}\right)\ket{b}\nonumber\\
=&\left(\sum_{x:f(x)=1}\frac{\alpha_{j,x}}{\sqrt{\nu_x}}\ket{\hat{0}}\right)+
\frac{1}{\sqrt{\C A}}\sum_{i\in[n], b\in\{0,1\}}\ket{i}\ket{\SC_{j,i,b}'}\ket{b},
\end{align}
and similarly,
\begin{equation}
\ket{\zeta_j}=\left(\sum_{x:f(x)=1}\frac{\alpha_{j,x}}{\sqrt{\nu_x}}\ket{\hat{0}}\right)+\frac{1}{\sqrt{\C A}}\sum_{i\in[n], b\in\{0,1\}}\ket{i}\ket{\SC_{j,i,b}}\ket{b}.\label{eq:zeta_rewrite}
\end{equation}

Then
\begin{align}
\braket{\zeta_j'}{\zeta_l'}
&=  \left(\sum_{x:f(x)=1} \frac{\alpha_{j,x}}{\sqrt\nu_x}\right)\left(\sum_{y:f(y)=1} \frac{\alpha_{l,y}}{\sqrt\nu_y}\right) 
+ \frac{1}{\C A} \sum_{i \in [n], b\in\{0,1\}} \braket{\SC_{j,i,b}'}{\SC_{l,i,b}'}.\label{eq:zeta_prime_inner}
\end{align}

Applying \cref{lem:JLcorr}, we have that
\begin{align}
  \sum_{i \in [n], b\in\{0,1\}} \braket{\SC_{j,i,b}'}{\SC_{l,i,b}'}
  &\in \sum_{i \in [n], b\in\{0,1\}} \left( \braket{\SC_{j,i,b}}{\SC_{l,i,b}}
  \pm 2\varepsilon (\| \ket{\SC_{j,i,b}} \|^2 + \| \ket{\SC_{l,i,b}} \|^2 ) \right) \\
  &\in \left(\sum_{i \in [n], b\in\{0,1\}} \braket{\SC_{j,i,b}}{\SC_{l,i,b}}\right)
  \pm 2\varepsilon \left(\sum_{i \in [n], b\in\{0,1\}}  \| \ket{\SC_{j,i,b}} \|^2 + \| \ket{\SC_{l,i,b}} \|^2 \right).\label{eq:zetajib_prime_inner}
\end{align}

Since $\ket{\zeta_j}$ has norm 1, we have
\begin{equation}
  1 = \| \ket{\zeta_j} \|^2 = 
  \left( \sum_{x: f(x) = 1} \frac{\alpha_{j,x}}{\sqrt{\nu_x}}\right)^2
  + \frac{1}{\C A} \sum_{i \in [n], b\in\{0,1\}} \| \ket{\SC_{j,i,b}} \|^2 \\
  \implies \sum_{i \in [n], b\in\{0,1\}} \| \ket{\SC_{j,i,b}} \|^2 \le
  \C A,
  \label{eq:zeta_norm_bound}
\end{equation}
where the last inequality follows because each $\alpha_{j,x}$ is a real number.

Plugging \cref{eq:zeta_norm_bound} into \cref{eq:zetajib_prime_inner}, we have that
\begin{align}
  \sum_{i \in [n], b\in\{0,1\}} \braket{\SC_{j,i,b}'}{\SC_{l,i,b}'}
  &\in \left(\sum_{i \in [n], b\in\{0,1\}} \braket{\SC_{j,i,b}}{\SC_{l,i,b}}\right)
  \pm 4\varepsilon \C A. \label{eq:zeta_inner_bound}
\end{align}

Finally, plugging \cref{eq:zeta_inner_bound} into \cref{eq:zeta_prime_inner} and 
then using \cref{eq:zeta_rewrite}, we have
\begin{align}
\braket{\zeta_j'}{\zeta_l'}
&\in \left(\sum_{x:f(x)=1} \frac{\alpha_{j,x}}{\sqrt\nu_x}\right)\left(\sum_{y:f(y)=1} \frac{\alpha_{l,y}}{\sqrt\nu_y}\right) 
+ \left(\frac{1}{\C A} \sum_{i \in [n], b\in\{0,1\}} \braket{\SC_{j,i,b}}{\SC_{l,i,b}} \right)
+ 4\varepsilon \tag{by \cref{eq:zeta_inner_bound}} \\
&\in \braket{\zeta_j}{\zeta_l} \pm 4\varepsilon \nonumber \\
&\in \delta_{i,j}\pm 4\varepsilon.
\end{align}
where in the final
expression we used that $\{\ket{\zeta_j}\}$ form an orthonormal basis.
\end{proof}

\begin{proof}[Proof of \cref{lem:compressed_relations_update} \cref{part:zeta_phi}]

Observe that for $x\in f^{-1}(0),$
\begin{align}
  1 = \| \ket{\phi_x} \|^2 = 
  \frac{1}{\mu_x} \left( 1 + \C A \sum_{i \in [n]}
  \| \ket{v_{x,i} }\|^2 \right)
  \implies \sum_{i \in [n]} \| \ket{v_{x,i}} \|^2
  = \frac{\mu_x-1}{\C A} \le A
  \label{eq:vxi_norm_bound}
\end{align}
where the last inequality follows by 
\cref{eq:mu_x_bound}.

Then using \cref{lem:JLcorr} and our choice of compression vectors $C$ (\cref{eq:compression_vectors}), we have
\begin{align}
  \sum_{i \in [n]}
  \braket{\SC_{j,i,b}'}{v_{x,i}'}
  &\in \sum_{i \in [n]} (\braket{\SC_{j,i,b}}{v_{x,i}}
  \pm 2\varepsilon (\| \SC_{j,i,b} \|^2 + \| v_{x,i} \|^2)) \nonumber\\
  &\in \left(\sum_{i \in [n]} \braket{\SC_{j,i,b}}{v_{x,i}}\right)
  \pm \left(2\varepsilon \sum_{i \in n} \| \SC_{j,i,b} \|^2 +
  2\varepsilon \sum_{i \in n} \| v_{x,i} \|^2 \right)\nonumber\\ 
  &\in \left(\sum_{i \in [n]} \braket{\SC_{j,i,b}}{v_{x,i}} \right)
  \pm 2\varepsilon (\C+1) A,
  \label{eq:zeta_vxi_bound}
\end{align}
where the last inequality follows by 
\cref{eq:zeta_norm_bound} and \cref{eq:vxi_norm_bound}.

From \cref{eq:zetap_in_zeta}, we have
\begin{align}
\braket{\zeta_j'}{\phi_x'} &= \left(\left(\sum_{y:f(y)=1}  \frac{\alpha_{j,y}}{\sqrt{\nu_y}}\bra{\hat{0}}\right)
+ \frac{1}{\sqrt{\C A}}\sum_{i\in[n], b\in \{0,1\}}\bra{i} \bra{\SC_{j,i,b}'}\bra{b} \right)\nonumber\\
&\qquad \left(\frac{1}{\sqrt{\mu_x}}\left(\ket{\hat{0}}-\sqrt{\C A}\sum_{i\in[n]}\ket{i}\ket{v_{x,i}'}\ket{\bar{x}_i}\right)\right) \nonumber\\
&= \sum_{y:f(y)=1}  \frac{\alpha_{j,y}}{\sqrt{\mu_x\nu_y}} +  \frac{1}{\sqrt{\mu_x}} \sum_{i\in[n]}\braket{\SC_{j,i,\bar{x}_i}'}{v_{x,i}'}\nonumber\\
&\in \left(\sum_{y:f(y)=1}  \frac{\alpha_{j,y}}{\sqrt{\mu_x\nu_y}} 
+ \frac{1}{\sqrt{\mu_x}} \sum_{i\in[n]}\braket{\SC_{j,i,\bar{x}_i}}{v_{x,i}}\right) \pm \frac{2 \varepsilon (\C+1) A}{\sqrt{\mu_x}} \tag{by \cref{eq:zeta_vxi_bound}} \\
&\in \braket{\zeta_j}{\phi_x} \pm 2\varepsilon(\C+1) A, 
\end{align}
where we have used that $\mu_x\geq 1$ (see \cref{eq:mu_x_bound}). Now because $\{\ket{\zeta_j}\}$ is an orthonormal
basis vector of $\textrm{span}\{\ket{\psi_x : f(x)=1} \}$, and $\ket{\phi_x}$ for $x\in f^{-1}(0)$ is orthogonal
to all $\ket{\psi_y}$ for $y\in f^{-1}(1)$, we have $\braket{\zeta_j}{\phi_x}=0$.
Thus $|\braket{\zeta_j'}{\phi_x'}| \leq 2\varepsilon(\C+1) A$.

\end{proof}

\begin{proof}[Proof of \cref{lem:compressed_relations_update} \cref{part:psiphiNormBounds}]
We first analyze $\|\ket{\psi_x'}\|.$
Because $\{\ket{\zeta_{j}}\}_{j\in[\kap]}$ form an orthonormal basis for 
$\textrm{span}\{ \ket{\psi_x} : f(x) = 1\}$,
there are real numbers $\hat{\alpha}_{x,j}$ such that
\begin{align}
\ket{\psi_x}=\sum_{j=1}\hat{\alpha}_{x,j}\ket{\zeta_j}, \quad \textrm{and}\quad
\ket{\psi_x'}=\sum_{j=1}\hat{\alpha}_{x,j}\ket{\zeta_j'}. \label{eq:psi_in_terms_of_zeta}
\end{align}
Applying \cref{part:zeta_preserve}, we have
\begin{align}
\|\ket{\psi_x'}\|^2 &= \sum_{j,l \in [\kap]} \hat{\alpha}_{x,j} \hat{\alpha}_{x,l} \braket{\zeta_j'}{\zeta_{l}'}
\in \sum_{j,l \in [\kap]} \hat{\alpha}_{x,j} \hat{\alpha}_{x,l} (\braket{\zeta_j}{\zeta_{l}}\pm 4\varepsilon) \\
&\in \|\ket{\psi_x}\|^2 \pm4\varepsilon\left(\sum_{j\in[\kap]}\hat{\alpha}_{x,j}\right)^2. \label{eq:psi_norm1}
\end{align}
Now
\begin{align}
\sum_{j\in[\kap]}\hat{\alpha}_{x,j}\leq \sqrt{\kap \sum_{j\in[\kap]}\hat{\alpha}_{x,j}^2}=\sqrt{\kap},
\label{eq:psi_norm2}
\end{align}
where we have used Cauchy-Schwarz and from \cref{eq:zetaj} that $1=\braket{\psi_x}{\psi_x}=\sum_{j\in[\kap]}\hat{\alpha}_{x,j}^2$. 
Plugging \cref{eq:psi_norm2} into \cref{eq:psi_norm1} and using that $\|\ket{\psi_x}\|^2=1$ gives us the desired result.

Next we analyze $\|\ket{\phi_x'}\|$.
We have
\begin{align}
\|\ket{\phi_x'}\|^2&=\left(\frac{1}{\sqrt{\mu_x}}\left(\bra{\hat{0}}-\sqrt{\C A}\sum_{i\in[n]}\bra{i}(\bra{v_{x,i}}S^\dagger)\bra{\bar{x}_i}\right)\right)
\left(\frac{1}{\sqrt{\mu_x}}\left(\ket{\hat{0}}-\sqrt{\C A}\sum_{j\in[n]}\ket{j}S\ket{v_{x,j}}\ket{\bar{x}_j}\right)\right)\\
&=\frac{1}{\mu_x}+\C A\sum_{i\in [n]}\bra{v_{x,i}}S^\dagger S\ket{v_{x,i}}
\end{align}
Since $\{\ket{v_{x,i}}\}_{x:f(x)=0}$ are part of the compressed vector set from
\cref{lem:JLcorr}, we have
\begin{align}
\|\ket{\phi_x'}\|^2&\in\frac{1}{\mu_x}\left(1+\C A\sum_{i\in [n]}\left(\braket{v_{x,i}}{v_{x,i}}\pm3\varepsilon\|\ket{v_{x,i}}\|^2\right)\right).
\end{align}
Thus, using \cref{eq:mu_x_bound}
\begin{align}
\left|\|\ket{\phi_x'}\|^2-\|\ket{\phi_x}\|^2\right|=\left|\|\ket{\phi_x'}\|^2-1\right|&\leq
\frac{3 \C \varepsilon A\sum_{i\in [n]}\|\ket{v_{x,i}}\|^2}{\mu_x}\leq 3\varepsilon
\end{align}
where we use $\mu_x = 1 + cA \sum_{i\in [n]}\|\ket{v_{x,i}}\|^2$ in the final inequality.
\end{proof}

We need one final tool before we can prove our main JL compression result, \cref{lem:deltaPhiAnalysis}. Our algorithm applies a reflection over the space
spanned by $\{\ket{\zeta_j'}\}_{j\in[\kap]}$. It would be convenient for our analysis if $\{\ket{\zeta_j'}\}_{j\in[\kap]}$ were an orthonormal basis for this space, but because of the JL compression, these vectors are only approximately orthonormal. However the following lemma allows us to create an orthonormal basis $\{\ket{v_j}\}_{j\in[\kap]}$ for the $\textrm{span}\{\ket{\zeta_j'}\}_{j\in[\kap]}$ that has high overlap with $\{\ket{\zeta_j}\}_{j\in[\kap]}$.
\begin{restatable}{lemma}{createOrtho}\label{lemma:ortho_basis}
Consider a set of vectors $\{\ket{\zeta_j}\}_{j\in [r]}$.
Fix $\varepsilon > 0$ so that $\varepsilon r < 1/4$
and $|\braket{\zeta_j}{\zeta_i}-\delta_{i,j}| \leq \varepsilon$ for all $i,j\in[r]$.
Then we can construct an orthonormal set $\{\ket{\ortho_j}\}_{j \in [r]}$ 
that is close to $\{\ket{\zeta_j}\}_{j \in [r]}$ in the sense that 
$\ket{\ortho_j} = \sum_{i=1}^r a_{j, i} \ket{\zeta_i}$ where $|a_{j,i}-\delta_{i,j}|\leq 3\varepsilon$
for all $j \in [r]$.
\end{restatable}

\begin{proof}[Proof of \cref{lemma:ortho_basis}]
We will first create unnormalized orthogonal vectors $\{\ket{v'_j}\}$ using Gram-Schmidt decomposition, and then we will
normalize them. Recall that Gram-Schmidt orthogonal but unnormalized vectors are created as follows. For $j\in [r]$, define
\begin{equation}\label{eq:GramSchmidt}
\ket{\ortho'_j}\coloneqq \ket{\zeta_j}-\sum_{k=1}^{j-1}\proj{\ortho'_k}\ket{\zeta_j}.
\end{equation}

We will prove using strong induction that $\forall n\in [r]$,
$\ket{\ortho'_n} = \sum_{j=1}^{n} a'_{n, j} \ket{\zeta_j}$ where $|a'_{j,i}|\leq 2\varepsilon$ if $j \ne i$ and $a'_{n,n} = 1$.

For the base case, we have $\ket{\ortho'_1}=\ket{\zeta_1}$, as desired. 

For the inductive step, assume that for
vectors $\ket{\ortho'_j}$ with $j\leq n$, our inductive assumption holds. Now consider $\ket{\ortho'_{n+1}}$:
\begin{align}
\ket{\ortho'_{n+1}} &\coloneqq \ket{\zeta_{n+1}} - \sum_{j=1}^n \ketbra{\ortho'_j}{\ortho'_j}\ket{\zeta_{n+1}} \tag{\cref{eq:GramSchmidt}}\\
&= \ket{\zeta_{n+1}} - \sum_{j=1}^n \left(\sum_{k=1}^{j} a'_{j, k} \ket{\zeta_k} \right)\left( \sum_{i=1}^{j} a'_{j, i} \braket{\zeta_i}{\zeta_{n+1}} \right) \tag{inductive assumption} \\
&= \ket{\zeta_{n+1}} - \sum_{k=1}^n\left(\sum_{j=k}^n\sum_{i=1}^ja'_{j, k}  a'_{j, i} \braket{\zeta_i}{\zeta_{n+1}} \right)\ket{\zeta_k} \tag{reordering summation}
\end{align}
Thus we have $a'_{n+1,n+1}=1$ and
\begin{align}
|a'_{n+1,k}|&=\left|\sum_{j=k}^n\sum_{i=1}^ja'_{j, k}  a'_{j, i} \braket{\zeta_i}{\zeta_{n+1}}\right|\nonumber\\
&\leq\varepsilon\sum_{j=k}^n\sum_{i=1}^j|a'_{j, k}|  \left|a'_{j, i}\right|\tag{lemma assumption on $\{\zeta_j\}$}\\
&\leq\varepsilon\sum_{j=k}^n\left|a'_{j, k} \right| (1+2j\varepsilon)\tag{inductive assumption}\\
&\leq\varepsilon\left((1+2\varepsilon k)+2\varepsilon\sum_{j=k+1}^n(1+2j\varepsilon)  \right)\tag{inductive assumption}\\
&\leq\varepsilon\left( 1 + 2\varepsilon k + 2\varepsilon (n-k) + 4\varepsilon^2 n^2  \right)\tag{upper bound sum} \\
&\leq\varepsilon\left( 1+ 2\varepsilon n (1 + 2\varepsilon n)  \right)\nonumber\\
&\leq\varepsilon \left(1 + \frac{1}{2} \left(1+\frac{1}{2}\right) \right) \le 2\varepsilon,
\end{align}
where in the last line, we have used our assumption that $\varepsilon n \le \varepsilon r \le \frac{1}{4}$.

Now by construction, $\{\ket{\ortho_j'}\}$ are orthogonal, so all that remains is to normalize them, setting $\ket{\ortho_j} = \frac{1}{\|\ket{\ortho'_j}\|} \ket{\ortho'_j}$. Writing
the normalized vectors as $\ket{\ortho_j}=\sum_{i=1}^ja_{j,i}\ket{\zeta_i}$, we find the coefficients scale as
\begin{align}
a_{j,i}=\frac{a'_{j,i}}{\|\ket{\ortho'_j}\|}.
\end{align}

Since
\begin{equation}
\|\ket{\ortho'_j}\|^2=\left(\bra{\zeta_j}+\sum_{i=1}^{j-1}a'_{j,i}\bra{\zeta_i}\right)
\left(\ket{\zeta_j}+\sum_{l=1}^{j-1}a'_{j,l}\ket{\zeta_l}\right)
\end{equation}
and $\braket{\zeta_j}{\zeta_j} \in 1\pm \varepsilon$, we have
\begin{align}
\left|\|\ket{\ortho'_j}\|^2-1\right|&\leq \varepsilon+2\sum_{i=1}^{j-1}|\braket{\zeta_j}{\zeta_i}||a'_{j,i}|+\sum_{i,l=1}^{j-1}|\braket{\zeta_i}{\zeta_l}||a'_{j,i}||a'_{j,l}|\nonumber\\
&\leq \varepsilon+4\varepsilon^2 j+4\varepsilon^3j^2\nonumber\\
&\leq \varepsilon + \frac{1}{4r} + \frac{1}{16r} \leq 2\varepsilon,
\end{align}
since $j \le r$ and $\varepsilon \leq \frac{1}{4r}$.
Thus
\begin{equation}\label{eq:vpj_norm_bound}
\sqrt{1- 2\varepsilon}\leq \|\ket{\ortho'_j}\|\leq \sqrt{1+ 2\varepsilon}.
\end{equation}

We now show that $|a_{j,j} - 1| \leq 2 \varepsilon$.
For the first side, we 
combine \cref{eq:vpj_norm_bound} and the observation that 
$1/\sqrt{1-2x}-1 \leq 2x$ for $x \in [0,.25]$.
Then
\begin{align}
  a_{j,j} - 1 = \frac{1}{\| \ket{\ortho_j'} \|} - 1
  \leq \frac{1}{\sqrt{1-2\varepsilon}} -1 \leq 2\varepsilon.
\end{align}
For the other side, we combine \cref{eq:vpj_norm_bound}
and the observation that $1-1/\sqrt{1+2x} \leq x$ for
$x \geq 0$.
Then
\begin{align}
  1-a_{j,j} = 1- \frac{1}{\| \ket{\ortho_j'} \|}
  \leq 1-\frac{1}{\sqrt{1+2\varepsilon}} \leq x.
\end{align}
Together, we have $|a_{j,j}-1| \leq 2\varepsilon$.

We now show that $|a_{j,i}| \leq 3\varepsilon$ for
$i \neq j$.
We combine \cref{eq:vpj_norm_bound} and the observation
that $2x/\sqrt{1-2x} \leq 3x$ for $x \in [0,.25]$.
Then
\begin{equation}
|a_{j,i}|=\frac{|a'_{j,i}|}{\|\ket{\ortho'_j}\|}
\leq \frac{2\varepsilon}{\sqrt{1-2\varepsilon}} \leq 3\varepsilon.
\end{equation}

The statement of the lemma follows.

\end{proof}

Now we have the tools to prove \cref{lem:deltaPhiAnalysis}.

\deltaPhi*

\begin{proof}[Proof of \cref{lem:deltaPhiAnalysis}]
By \cref{lem:compressed_relations_update} \cref{part:zeta_preserve}, 
the Johnson Lindenstrauss compression of the vectors $\{\ket{\zeta_j}\}$
nearly preserves their orthonormality, i.e. $\forall j,l\in[\kap], |\braket{\zeta_j'}{\zeta_l'}-\delta_{j,l}|\leq 4\varepsilon$.
Then we can apply \cref{lemma:ortho_basis} to obtain an orthonormal set of vectors 
$\{\ket{\ortho_j}\}_{j\in[\kap]}$ such that
\begin{equation}
\ket{\ortho_j}=\sum_{j\in[\kap]}a_{j,i}\ket{\zeta_j'} \quad \forall j\in[\kap]
\end{equation}
where $|a_{j,i}-\delta_{j,i}|\leq 12 \varepsilon.$
We can write $\Delta'=\sum_{j\in[\kap]}\proj{\ortho_j}.$
Then
\begin{align}
\|\Delta'\ket{\phi_x'}\|&= \left\|\sum_{j \in [\kap]} \ketbra{\ortho_j}{\ortho_j} \ket{\phi_x'}\right\|
= \left\|\sum_{j \in [\kap]} \sum_{g\in [\kap]} a_{j,g}\ket{\ortho_j}\braket{\zeta_g'}{\phi_x'}\right\|\nonumber\\
&\in 2\varepsilon(\C+1)A \left\|\sum_{j \in [\kap]} \sum_{g\in [\kap]} 
(\delta_{j,g}\pm 12\varepsilon)\ket{\ortho_j}\right\|
\tag{by \cref{lem:compressed_relations_update} \cref{part:zeta_phi,lemma:ortho_basis}}\\
&\in 2\varepsilon(\C+1)A (1\pm \kap 12\varepsilon)\left\|\sum_{j \in [\kap]} 
\ket{\ortho_j}\right\| 
\nonumber\\
&\in 2\varepsilon(\C+1)A (1+\kap 12\varepsilon)\sqrt{\kap} 
\tag{Since $\{\ket{\ortho_j}\}$ are orthonormal}\\
&\in 4\varepsilon(\C+1)A \sqrt{\kap}
\end{align}
where we use that $\varepsilon \kap \leq 1/12$.

Thus
\begin{equation}
\|(R+I)\ket{\phi_x'}\|=\|(2\Delta'-I+I)\ket{\phi_x'}\|=8\varepsilon (\C+1) A \sqrt{\kap}.
\end{equation}
\end{proof}

In the following lemma, we show that the maximum rank
of an $f$-deciding vector set is equal to the rank of
the vectors $\{\ket{\psi_x}\}_{x\in f^{-1}(1)}$.
This allows us to related the compression of \cref{thm:exact_compress}, which depends on the maximum rank, 
to the compression of \cref{thm:jl_inexact}, which depends on 
the rank of $\{\ket{\psi_x}\}_{x\in f^{-1}(1)}.$
\begin{restatable}{lemma}{kappa_is_r}\label{lem:kappa_is_r}
Let $\kap'$ be the maximum rank of an $f$-deciding vector set 
$\{\ket{v_{x,i}}\}_{x\in X,i\in [n]}$.
Let $\kap$ be the rank of $\{\ket{\psi_x}:f(x)=1\}$,
where $\ket{\psi_x}$ is created from $\{\ket{v_{x,i}}\}$ as in \cref{eq:psidef}.
Then $\kap' \leq \kap \leq 2n \kap'$.
\end{restatable}

\begin{proof}[Proof of \cref{lem:kappa_is_r}]
Let $r_i$ be the rank of $V_i \coloneqq \{\ket{v_{x,i}} : f(x) = 1\}$.
We first show that for all $i\in [n]$, $\kap \ge r_i$, and therefore $\kap\geq \kap'$.

Fix $i \in [n]$. There exists some linearly independent set of vectors $A_V \subseteq V_i$ of size $r_i$. Each element of $A_V$ corresponds to a single input $x$, so we can define $A_x \coloneqq \{x : \ket{v_{x,i}} \in A_V \}$. Now consider the set $A_\psi \coloneqq \{\ket{\psi_x} : x \in A_x \}$. Assume for contradiction that $A_\psi$ is linearly dependent. Then there exists some scalar weights 
$b_x$ for $x\in A_x$
such that $\sum_{x \in A_x} b_x \ket{\psi_x} = 0$ and thus from \cref{eq:psidef}
\begin{equation}
0=\sum_{\substack{x \in A_x}}\frac{b_x}{\sqrt{\nu_x}}\left(\ket{\hat{0}}+\frac{1}{\sqrt{\C A}}
\sum_{j\in[n]}\ket{j}\ket{v_{x,j}}\ket{x_j}\right),
\end{equation}
which implies
\begin{equation}
0=\sum_{\substack{x \in A_x:x_i=1}} b_x \ket{v_{x,i}}=\sum_{\substack{x \in A_x:x_i=0}} b_x \ket{v_{x,i}} = \sum_{\substack{x \in A_x}} b_x \ket{v_{x,i}}.
\end{equation}
This is a contradiction, as it implies that the set $A_V$ is not linearly independent.
We can then say that $A_\psi$ is linearly independent and that $\kap \ge |A_\psi| = r_i$. Hence $\kap\geq \kap'.$

We now show that $\kap \le 2n \kap'$.
We will construct a basis for
$\{\ket{\psi_x}:f(x)=1\}$ of size at most $2n\kap'$.
Let $B = \{ \}$.
For each $i \in [n]$, let $B_i$
be a smallest orthonormal basis for $\{\ket{v_{x,i}} : f(x) = 1\}$.
By the definition of $r_i$, we know $|B_i| = r_i$.
For every vector $\ket{w} \in B_i$, add
$\ket{i}\ket{w}\ket{0}$ and
$\ket{i}\ket{w}\ket{1}$ to $B$.
By the definition of $\ket{\psi_x}$ in \cref{eq:psidef},
$B$ is clearly an orthonormal basis for
a space containing $\{\ket{\psi_x}:f(x)=1\}$.
Now $|B| = \sum_{i \in [n]} 2 r_i \leq 2 n \kap'$.
The lemma statement follows.

In fact, the above bound is the best we can hope for
as we can see by considering OR.
In this case, $\kap' = 1$ and $\kap = 2n$ (the construction is explained in Section 23.3 of \cite{ACNotes}).

\end{proof}

\section{Proofs of Numerical Results}\label{sec:numericalProofs}

\SvdPsi*

\begin{proof}[Proof of \cref{lemma:svd_psi}]
Let $M'$
be the $\kap^*$ rank approximation for $M$ created using the singular value decomposition. Now by the well known properties of the SVD \cite{blum2020foundations}, $M'$ is the matrix whose
rows are the corresponding vectors $\{\Delta'\ket{\psi_x}\}_{x\in X_1}$,
i.e the projections of the vectors onto the subspace $\Delta'$.
Thus if $\ket{\psi_x}$ is in the $j^\textrm{th}$ row of $M$, 
we have that 
\begin{equation}\label{eq:diff1}
\|\Delta'\ket{\psi_x}-\ket{\psi_x}\|=\|\bra{j}(M-M')\|
\end{equation}
where $\ket{j}$ is the $j^\textrm{th}$ standard basis vector.

Then
\begin{equation}\label{eq:diff2}
\|\bra{j}(M-M')\| \leq \|M-M'\|,
\end{equation}
where $\|M\|$ is the spectral norm of $M$:
\begin{equation}
\|M\|=\max_{\ket{v}:\|\ket{v}\|^2=1}\|M\ket{v}\|.
\end{equation}
It is again well known \cite{blum2020foundations} that
\begin{equation}\label{eq:diff3}
\|M-M'\|=s_{\kap^*+1}.
\end{equation}

Combining \cref{eq:diff1,eq:diff2,eq:diff3}, we have
\begin{equation}
\|(2 \Delta' - I)\ket{\psi_x} - \ket{\psi_x} \| = 2 \|\Delta'\ket{\psi_x}-\ket{\psi_x}\|\leq 2s_{\kap^*+1}.
\end{equation}

\end{proof}

\SvdPhi*

\begin{proof}[Proof of \cref{lemma:svd_phi}]
Starting from the bound on the spectral norm between $M$ and $M'$ \cite{blum2020foundations},
we have that for all $y\in Y$,
\begin{equation}\label{eq:size1}
\|M\ket{\phi_y}-M'\ket{\phi_y}\|\leq  s_{\kap^*+1}. 
\end{equation}
Now because $M$ is a matrix whose $|X_1|$ rows are $\{\ket{\psi_x}\}_{x\in X_1}$,
by \cref{eq:innerprod}, we have
\begin{equation}\label{eq:size2}
\|M\ket{\phi_y}\|\leq \varepsilon\sqrt{|X_1|}.
\end{equation}
Combining \cref{eq:size1,eq:size2}, and using the reverse triangle inequality, we have 
\begin{equation}
\|M'\ket{\phi_y}\|\leq  s_{\kap^*+1}+\varepsilon\sqrt{|X_1|}.
\end{equation}
Now, because the rows of $M'$ are in the span of $\Delta'$, we have
\begin{equation}\label{eq:size3}
\|M'\ket{\phi_y}\|=\|M'\Delta'\ket{\phi_y}\|\geq s_{\kap^*}\|\Delta'\ket{\phi_y}\|,
\end{equation}
where we have used the fact that $M'$ decomposes as $M=UDV^\dagger$, 
where $U$ and $V$ are unitaries, and hence do not contribute to the
$\ell_2$ norm, and $D$ is a diagonal positive matrix whose smallest value
is $s_{\kap^*}$. 

Combining \cref{eq:size2,eq:size3} and rearranging,
we have
\begin{equation}
\|\Delta'\ket{\phi_y}\|\leq \frac{ s_{\kap^*+1}+\varepsilon\sqrt{n_1}}{s_{\kap^*}}.
\end{equation}
\end{proof}

\end{document}